\documentclass[conference]{IEEEtran}
\IEEEoverridecommandlockouts

\usepackage{cite}
\usepackage{amsmath,amssymb,amsfonts}
\usepackage{amsthm}
\usepackage{mathtools}
\usepackage{algorithmic}
\usepackage[ruled,vlined]{algorithm2e}
\usepackage{graphicx}
\usepackage{textcomp}
\usepackage[table]{xcolor}
\usepackage{booktabs}
\usepackage{makecell}
\usepackage{multirow}
\usepackage{adjustbox}
\usepackage{tabularx}
\usepackage{array}
\usepackage{balance}
\usepackage[hidelinks]{hyperref}

\newcolumntype{Y}{>{\centering\arraybackslash}X}

\newtheorem{theorem}{Theorem}

\theoremstyle{definition}
\newtheorem{definition}{Definition}

\newcommand{\best}[1]{\textbf{\underline{#1}}}
\newcommand{\second}[1]{\textbf{#1}}
\newcommand{\Frameworkname}{LasRepair}

\def\BibTeX{{\rm B\kern-.05em{\sc i\kern-.025em b}\kern-.08em
    T\kern-.1667em\lower.7ex\hbox{E}\kern-.125emX}}

\usepackage{etoolbox}
\usepackage{ragged2e}
\usepackage{needspace}

\makeatletter

\patchcmd{\@makecaption}
  {{\normalfont\footnotesize\scshape #2}}
  {\parbox[t]{\hsize}{%
   \setlength{\parindent}{0pt}%
   \justifying
   \noindent\normalfont\footnotesize #2\par
 }}
  {}
  {\PackageWarning{caption-style}{Failed to patch table caption}}

\let\old@makecaption\@makecaption
\long\def\@makecaption#1#2{%
  \old@makecaption{\bfseries #1}{\bfseries #2}%
}

\makeatother

\begin{document}

\title{Collaborative Large and Small Language Models for Accurate and Scalable Data Repair}

\author{
\IEEEauthorblockN{Qian Chen, Jianwei Wang, Wenjie Zhang}
\IEEEauthorblockA{
University of New South Wales\\
Sydney, Australia\\
\{qian.chen6,jianwei.wang1,wenjie.zhang\}@unsw.edu.au}
}

\maketitle

\begin{abstract}
We study the problem of data repair, a key task in data cleaning that corrects erroneous entries in raw datasets to improve overall data quality.
Although recent data-driven methods, especially those based on large language models (LLMs), achieve remarkable performance, we observe that:
(i) they directly repair data in the raw and low-quality context, which may compromise learning signals, and
(ii) they directly use uncertain model outputs as repairs, potentially introducing unreliable corrections and compromising repair quality.
Motivated by the efficiency of small language models (SLMs) and the capabilities of LLMs, and aiming to address the above limitations, we propose \textbf{\Frameworkname}, a framework that collaborates \underline{\textbf{L}}arge \underline{\textbf{a}}nd \underline{\textbf{s}}mall language models for data \underline{\textbf{Repair}}.
\Frameworkname{} employs an LLM as an instructor, which selects a global repair context to guide the SLM. The SLM acts as a corrector, using the selected context to repair erroneous data more efficiently.
Moreover, to further improve context quality, we extend \Frameworkname{} to \Frameworkname{}+, which formulates data repair as an Expectation-Maximisation (EM) procedure that alternates between an E-step for updating the corrector parameters and an M-step for refining the repair context.
Furthermore, to mitigate model uncertainty, we propose \Frameworkname{}++, which uses column-calibrated model confidence to down-weight unreliable repaired rows when updating the corrector, thereby enhancing repair quality.
Theoretical analysis and empirical evaluation demonstrate the superiority of our methods. 
We theoretically prove the effectiveness of the EM-style procedure and the confidence-based weighting. 
Experiments on real-world datasets show that \Frameworkname{}++~ achieves an average F1-score improvement of 18.1\% over the strongest baseline.
Code is available at https://github.com/T-Lab/LasRepair.
\end{abstract}


\suppressfloats[t]

\begin{figure}[t]
  \vspace{15pt}
  \centering
  \includegraphics[width=\columnwidth, trim = 0 8 0 25]{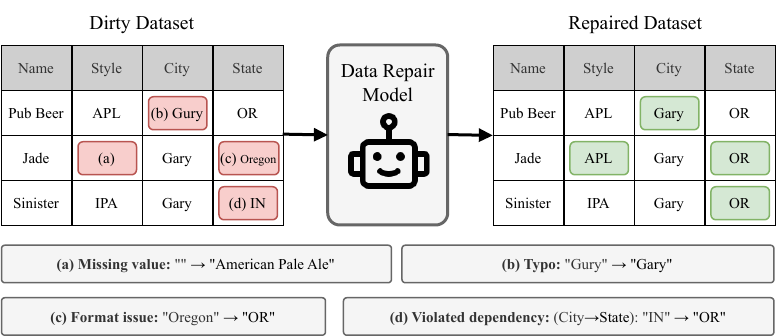}
  \caption{An illustrated example of common data errors. The dirty table contains (a)  a missing value $x_2[Style]$, (b) a typo in $x_2[City]$, (c) a formatting issue in $x_1[State]$, and (d) an attribute-dependency violation (e.g., $City\!\rightarrow\!State$) reflected in $x_3[State]$.}
  \label{fig:example}
\end{figure}

\begin{figure*}[htbp]
  \centering
  \includegraphics[width=\textwidth, trim= 0 15 0 8]{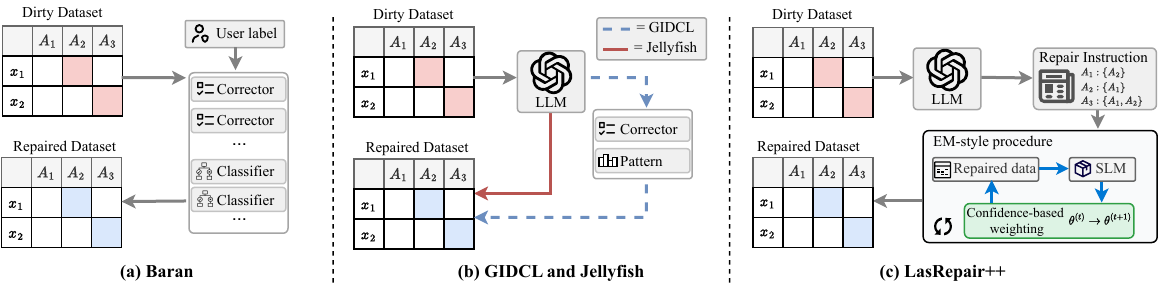}
  \caption{Conceptual comparison of representative data repair methods. (a) Baran relies on user-provided labels to train correctors and classifiers to produce a repaired dataset. (b) GIDCL and Jellyfish incorporate LLMs either to assist in correction pattern discovery or to perform direct repair. (c) \Frameworkname{}++ generates repair instructions with an LLM instructor and performs an EM-style procedure with an SLM corrector, further incorporating confidence-based weighting to improve performance.}
  \label{fig:comp}
\end{figure*}

\section{Introduction}
Real-world datasets are often noisy and error-prone, containing errors such as missing values, typos, and other inconsistencies, as illustrated in Figure~\ref{fig:example}. Such low-quality data can negatively affect the training of downstream models~\cite{li2021cleanml}, especially large language models (LLMs), which are sensitive to data quality~\cite{wang2025unimp, li2025datacomp}. To address this issue, data cleaning has been widely studied as a fundamental task in data quality, aiming to correct erroneous entries in relational tables and improve the reliability of real-world datasets~\cite{rahm2000dataclean, ilyas2019datacleanbook}.
Data cleaning typically consists of two tasks: error detection and data repair (i.e., error correction). 
Error detection aims to identify erroneous data entries, and data repair focuses on correcting them once they have been detected~\cite{mengqi2025treeed}. 
In this work, we focus on data repair and aim to improve data quality.


The problem of data repair has been studied extensively in the data management field (see the survey paper~\cite{wei2024survey} for a more comprehensive introduction).
Conventional data repair methods mainly include \emph{constraint-driven} methods~\cite{xu2013holistic, zuhair2015bigdansing, gao2019mlnclean,rezig2021horizon,dall2013nadeef} and \emph{hybrid} methods~\cite{theo2017holoclean,fei2011unified,geo2012relative} that leverage integrity constraints, statistical distributions, or downstream model performance to generate repair candidates and select repairs based on specific optimization criteria~\cite{bohan2007fd,theo2017holoclean,xu2013holistic}. 
However, these methods typically rely on handcrafted rules or task-specific assumptions, which limit their adaptability to complex, heterogeneous data and restrict their effectiveness when constraints are incomplete, noisy, or difficult to specify in practice.

Recently, \emph{data-driven} methods~\cite{mohamed2013scare,mohammad2020baran, yan2024gidcl,zhang2024jellyfish} have emerged as a more flexible alternative, learning to propose and select corrections directly from data and achieving state-of-the-art (SOTA) performance.
These methods can be broadly categorized into discriminative classifier-based methods and generative LLM-based methods.
Discriminative methods, exemplified by Baran~\cite{mohammad2020baran} in Figure~\ref{fig:comp}(a), generate repair candidates from multiple contextual views and train lightweight classifiers to select plausible corrections with limited supervision.
Generative LLM-based methods, exemplified by Jellyfish~\cite{zhang2024jellyfish} and GIDCL~\cite{yan2024gidcl} in Figure~\ref{fig:comp}(b), exploit the semantic capabilities and world knowledge of LLMs for generative data repair.

\noindent\textbf{Motivation.} 
Although these data-driven methods achieve SOTA accuracy, two crucial limitations persist.

First, existing repair methods are affected by \emph{context contamination}.
These methods are trained on raw datasets that may contain an unavoidable fraction of errors.
As a result, the context used to repair one cell can contain misleading information.
This issue is particularly severe when the error rate is high or when attributes are strongly correlated.
To mitigate this problem, Baran~\cite{mohammad2020baran} relies on user-labeled examples as clean supervision. However, such labels are costly and typically limited in scale, making it difficult to construct consistently reliable repair contexts.

Second, they may directly use uncertain model outputs as repairs. 
Such uncertainty can arise across both model training and inference.
During training, these repair models often use stochastic algorithms to optimize the loss, such as SGD~\cite{mandt2018sgd}, which optimizes the loss through randomized gradient updates, introducing uncertainty into the learned parameters.
For inference, these repair models usually generate probability distributions over possible outputs and commit to high-probability candidates, discarding the remaining probability mass.
For example, Jellyfish~\cite{zhang2024jellyfish} repairs serialized instances through autoregressive generation, where uncertainty in token-level probability distributions can accumulate across decoding steps and lead to uncertain repaired values in the final output~\cite{malinin2021uncertainty}.

\noindent\textbf{Challenges.} To overcome these limitations and design an effective data repair method, two main challenges exist.

\textit{Challenge I: How to find a reliable context in the low-quality table with complex relationships?} 
Real-world tables often contain various types of errors across multiple attributes that contaminate the repair context, making reliable context selection difficult.
A straightforward strategy is to remove all erroneous values, but such aggressive filtering can discard informative attributes or tuples, leading to insufficient context and biased repairs~\cite{Wang2023listwise}.
Another promising direction is to select correlated attributes as context to balance error reduction and information preservation~\cite{yan2024gidcl,guyon2003selection,yu2004effecient, shiyuan2026hyperjoin}.
Such a trade-off is difficult to achieve in complex relational tables, where useful dependencies may be implicit and entangled with corrupted values~\cite{theo2017holoclean}.
Therefore, finding reliable context from a low-quality table remains a challenge.

\textit{Challenge II: How to effectively and efficiently alleviate model uncertainty on large datasets?}
On large datasets, alleviating model uncertainty is challenging in terms of both effectiveness and efficiency~\cite{ovadia2019trustmodel, jianwei2024nomi}. 
On the one hand, effectiveness is hard to guarantee because models trained on large datasets can exhibit high prediction variance, leading to unstable and inconsistent repair outputs~\cite{mandt2018sgd,damour2020underspecification}.
On the other hand, maintaining efficiency at scale is challenging because the number of cells requiring repair grows with dataset size, making even lightweight repair operations computationally costly.
Therefore, effectively reducing uncertainty while avoiding additional verification or external checking remains a key challenge for large-scale data repair.


\noindent\textbf{Our approaches.} 
Motivated by the above challenges and the observation that small language models (SLMs) offer efficiency and LLMs provide strong reasoning capabilities, we propose \Frameworkname{}, a \underline{\textbf{L}}arge \underline{\textbf{a}}nd \underline{\textbf{s}}mall language model for accurate and scalable data \underline{\textbf{Repair}}. 
\Frameworkname{} adopts the \textit{LLM-as-an-instructor} paradigm~\cite{zheng2023judging}, in which an LLM instructs a corrector (i.e., an SLM) to enable scalable data repair.
Specifically, instead of relying on an LLM to repair every erroneous cell, \Frameworkname{} leverages the strong semantic reasoning capabilities and broad prior knowledge of the LLM~\cite{brown2020gpt3} to derive compact and relevant repair contexts for each target column.
Based on these contexts, the corrector is fine-tuned as a cell-level sequence-to-sequence repair model.

Based on \Frameworkname{} and to further address \textit{Challenge I}, we propose \Frameworkname{}+. 
\Frameworkname{}+ formulates data repair as an Expectation--Maximization (EM) process~\cite{apd1979em, Neal1998emview, AP1977EM} and iteratively performs the maximization step (M-step) and expectation step (E-step) to progressively improve the context quality and repair accuracy~\cite{xie2020selftraining}.
In the M-step, the corrector infers candidate repairs for the dataset and writes the generated repair back to the dataset to improve context quality.
In the E-step, the updated context is used as input to further fine-tune the corrector, enabling it to learn from increasingly reliable repaired data.
The updated corrector subsequently produces updated repairs for the next M-step.


Building on \Frameworkname{}+, we introduce \Frameworkname{}++ to address \textit{Challenge II} by assigning lower weights to potentially low-quality repairs during the E-step.
Rather than treating all generated repairs equally, \Frameworkname{}++ estimates the reliability of generated repairs and incorporates this reliability into model updates~\cite{ren2019learningreweight}.
Specifically, it derives confidence-based weights from model predictions and uses them to mitigate model uncertainty.
This lightweight weighting process reuses existing model outputs without requiring additional model calls, thereby avoiding substantial computational overhead.
Consequently, repairs with lower weights (higher uncertainty) contribute less to corrector updates.

\noindent\textbf{Theoretical and empirical studies.} We theoretically and empirically demonstrate the superiority of our methods. 
Theoretically, we analyze the proposed refinement procedure and show the monotonic behavior of the corresponding objective.
Moreover, we analyze how confidence-based weighting can reduce the influence of noise under a heteroscedastic noise model.
Empirical evaluations on 7 real-world datasets also demonstrate the strong performance of \Frameworkname{}++ in terms of 3 metrics.
\Frameworkname{}++ achieves F1-score improvement ranging from 5.6\% to 26.8\% over the previous SOTA methods. 
\Frameworkname{}++ also achieves the highest error drop rate on 6 out of 7 datasets and yields an average gain of 16\% over the best baseline.
Furthermore, \Frameworkname{}++ can efficiently handle large-scale datasets without compromising accuracy.

\noindent\textbf{Contributions.} The key contributions of this paper are summarized as follows.
\begin{itemize}
    \item We present the \Frameworkname{} family for data repair. The base \Frameworkname{} separates global context selection from local value generation, using an LLM as a dataset-level \emph{instructor} and a fine-tuned SLM as a cell-level \emph{corrector}.
    
    \item We propose an EM-style procedure that jointly updates the repaired data and the SLM corrector, allowing contaminated contexts to be progressively improved.
    \item We develop a confidence-based weighting mechanism that estimates the reliability of generated repairs and incorporates row-level confidence weights into SLM training, reducing the influence of model uncertainty.
    \item We demonstrate the effectiveness of \Frameworkname{}++ empirically through extensive experiments on 7 benchmark datasets, where it consistently outperforms SOTA baselines, yielding significant F1-score improvements and strong robustness across varying noise levels.
\end{itemize}

\section{PRELIMINARIES}
In this section, we formally define the problem and then review the SOTA solutions and their limitations. Frequently used notations are summarized in Table~\ref{tab:symbols}.

\subsection{Problem definition}
\begin{table}[t]
\centering
\small
\setlength{\tabcolsep}{3pt}
\renewcommand{\arraystretch}{1.0}
\caption{Summary of the main symbols and definitions used in this paper.
}
\label{tab:symbols}
\begin{tabular}{|p{0.24\columnwidth}|p{0.64\columnwidth}|}
\hline
\rowcolor{gray!20}
\textbf{Notation} & \textbf{Description} \\
\hline

$\mathbf{D}_e, \mathbf{D}_c, \mathbf{D}_r^{(t)}$
& Dirty/clean/repaired dataset. \\

$x_i, x_{ij}$
& Tuple $i$ and cell $(i,j)$. \\

$z_{ij}, \hat{z}_{ij}^{(t)}$
& Clean/repaired value at iteration $t$. \\

$\mathcal{E}$
& Set of detected erroneous cells. \\

$C$
& Column set. \\

$S$
& Dataset sketch for the instructor LLM. \\

$\mathbf{W}$
& Column influence matrix. \\

$\mathcal{G}$
& Induced column graph. \\

$\mathcal{T}_j$
& Influence tree for target column $j$. \\

$N(j)$
& Selected columns for target column $j$. \\

$\mathbf{x}_i^{(j)}$
& Serialized input for repairing cell $(i,j)$. \\

$g_{\theta}$
& SLM repair model. \\

$\theta^{(t)}$
& Model parameters at iteration $t$. \\

$\omega_i^{(t)}$
& Row-level weight at iteration $t$. \\

$\omega_i^{*(t)}$
& Normalized tuple weight at iteration $t$. \\

\hline
\end{tabular}
\end{table}

Following prior work~\cite{mohammad2020baran, wei2024survey, yan2024gidcl}, we consider a relational table with $n$ records and $d$ attributes. 
Let $\mathcal{X}_j$ denote the domain of attribute $j$, and let $\mathcal{X}=\mathcal{X}_1\times\cdots\times\mathcal{X}_d$ denote the tuple (record) domain.
The observed dirty table is $\mathbf{D}_e\in\mathcal{X}^{n}$ and the (unknown) clean ground-truth table is $\mathbf{D}_c\in\mathcal{X}^{n}$. Let $\mathbf{x}_i=(x_{i1},\ldots,x_{id})$ be the $i$-th record in $\mathbf{D}_e$, where $x_{ij}$ denotes the value of cell $(i,j)$ and $z_{ij}$ denotes its ground-truth value. For readability, we use column names $\mathcal{A}=\{a_1,\ldots,a_d\}$ and write $x_i[a_j]\triangleq x_{ij}$ (similarly $z_i[a_j]\triangleq z_{ij}$). 
For iterative repair, let $\mathbf{D}_r^{(t)}$ be the repaired table at iteration $t$ with $\mathbf{D}_r^{(0)}=\mathbf{D}_e$, and let $\mathbf{D}_r^{(T)}$ denote the final repaired table.

A cell $(i,j)$ is \emph{erroneous} if $x_{ij}\neq z_{ij}$. The unknown true error set is denoted by
$\mathcal{E}^\ast=\{(i,j)\mid x_{ij}\neq z_{ij}\}.$

In practice, $\mathbf{D}_c$ is unavailable, and an error detector returns a detected error set $\mathcal{E}\subseteq[n]\times[d]$, which specifies the cells to be repaired. This work focuses on the repair stage and assumes that $\mathcal{E}$ is given by an oracle detector.

\begin{definition}[Data repair]
Given a dirty table $\mathbf{D}_e$ and a detected error set $\mathcal{E}$, the goal of data repair is to produce a repaired table $\mathbf{D}_r$ that modifies only the detected erroneous cells and restores the erroneous values as closely as possible to their unknown clean values:
$$
\mathbf{D}_r[i,j]=x_{ij}\ \forall (i,j)\notin\mathcal{E},\quad
\min_{\mathbf{D}_r}\sum_{i=1}^{n}\sum_{j=1}^{d}\mathbb{I}\!\left[\mathbf{D}_r[i,j]\neq \mathbf{D}_c[i,j]\right].
$$
\end{definition}

\subsection{State-of-the-art}
As indicated by prior work and recent surveys~\cite{wei2024survey}, SOTA data-repair performance is increasingly achieved by data-driven methods that learn repair decisions from data rather than relying solely on manually specified constraints. 
Representative methods can be broadly grouped into two categories: 
(i) \emph{discriminative, classifier-based} methods that select repairs from a candidate set (e.g., Baran), and
(ii) \emph{generative, LLM-based} methods that directly generate repairs using the capabilities of LLMs (e.g., GIDCL and Jellyfish).
For simplicity, we refer to them as discriminative and generative methods.

\noindent\textbf{Discriminative methods.}
Baran~\cite{mohammad2020baran} is a representative discriminative repair framework. 
Given a dirty table and a set of detected erroneous cells, Baran instantiates multiple correctors that exploit complementary contextual signals of a cell to propose candidate repairs.
It then represents each cell-candidate pair with aggregated features and trains lightweight classifiers from limited user-provided corrections to select the most plausible candidate.

\noindent\textbf{Generative methods.}
Recent generative methods formulate repair as a sequence generation problem.
GIDCL~\cite{yan2024gidcl} combines graph-based table modeling with LLM-based repair, using structural signals to assist correction pattern discovery and improve repair consistency.
Jellyfish~\cite{zhang2024jellyfish} formulates data repair as an instruction-following generation task, where a serialized table context and error information are provided to a generative model to produce corrected values.

\section{OUR METHOD}
\begin{figure*}[ht]
  \centering
  \includegraphics[width=\textwidth, trim=0 5 0 5]{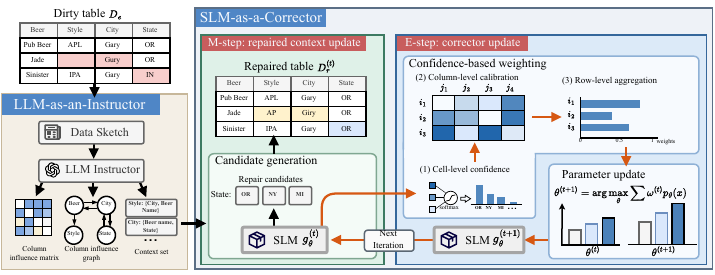}
  \caption{Overview of \Frameworkname{}++. Starting from the dirty table, the LLM instructor constructs a column influence graph and extracts the target-specific context set as repair instructions. The SLM corrector then generates repair candidates and the cell-level confidence. A confidence matrix is then computed and aggregated into confidence-based weights, which are used to update the model parameters. The process terminates when the stop criterion is satisfied and outputs the final repaired table.}
  \label{fig:workflow}
\end{figure*}

In this section, we introduce the \Frameworkname{} family, which consists of three progressively enhanced variants.
\Frameworkname{} is the base instructor--corrector framework in which an LLM instructor provides table-level structural guidance and an SLM corrector performs scalable cell-level repair.
\Frameworkname{}+ extends \Frameworkname{} with an EM-style procedure that repeatedly updates the repaired table and the corrector. 
\Frameworkname{}++ further extends \Frameworkname{}+ with confidence-based weighting to reduce the influence of model uncertainty.

\subsection{\Frameworkname{} method}

\noindent\textbf{Motivation.} 
A single model that directly repairs all cells often faces a trade-off between global reasoning and scalability.
This trade-off can be addressed by leveraging language models.
LLMs provide strong semantic reasoning, but applying an LLM to every erroneous cell can be costly, especially on large tables.
SLMs are more efficient and easier to fine-tune, but they require compact and informative context.
\Frameworkname{} therefore adopts an instructor--corrector paradigm: the instructor provides global, table-level structural guidance, while the corrector performs efficient cell-level repair.

\Frameworkname{} consists of two tightly connected components:
(i) an instructor model that infers a column influence structure to select a context set $\mathcal{N}(j)$ for each target column $j$;
(ii) a corrector model $g_\theta$ that repairs erroneous cells in $\mathcal{E}$ conditioned on the selected context.
This design separates global context selection from local value generation, reducing irrelevant or weakly informative context while preserving scalability.

\noindent \textbf{LLM-as-an-instructor.}
Structural decisions in \Frameworkname{} are made by an instructor at the table level.
Let $\mathcal{C}=\{1,\ldots,d\}$ denote the set of column indices.
We construct a compact table sketch $\mathbf{S}$ consisting of column names and a small set of rows sampled from $\mathbf{D}_e$, and then query the instructor once to infer inter-column informativeness.
The instructor returns a column influence matrix $\mathbf{W}\in[0,1]^{d\times d}$ with entries
\[
\mathbf{W}_{ab} \;=\; f_{\text{LLM}}(a,b,\mathbf{S}) \in [0,1], \qquad a\neq b,
\]
where $\mathbf{W}_{ab}$ quantifies the informativeness of column $a$ when repairing column $b$.
We interpret $\mathbf{W}$ as a directed weighted column influence graph~\cite{jianwei2024alice}:
\[
\mathcal{G} = (\mathcal{V}, \mathcal{A}, \mathbf{W}),\quad
\mathcal{V}=\mathcal{C},\quad
\mathcal{A}=\{(a\to b): a,b\in\mathcal{V}, a\neq b\}.
\]

Although $\mathbf{W}$ is defined pairwise, repairing a target column may require context from multiple correlated columns~\cite{jianwei2024transzero}.
We therefore aggregate influence along directed paths that lead to the target column.
For each target column $j\in\mathcal{C}$, and each candidate context column $v\neq j$, we define
\[
\phi_j(v)=
\max_{P:\, v\leadsto j,\, |P|\le h}
\prod_{(u\to u')\in P} \mathbf{W}_{u u'},
\]
where $P$ ranges over directed paths from $v$ to $j$ with length at most $h$.
Intuitively, $\phi_j(v)$ measures the strongest path through which column $v$ can inform the repair of column $j$, either directly or indirectly.
Given a threshold $\epsilon>0$, we then select the context set as
$$
\mathcal{N}(j)=\{\, v\in\mathcal{V}\setminus\{j\}:\ \phi_j(v)\ge \epsilon \,\}.
$$
By default, we set $\epsilon=0.6$ to control the context size. 
For interpretability, each selected $v\in\mathcal{N}(j)$ is connected to $j$ through the strongest path $\phi_j(v)$, resulting in an $h$-hop influence tree rooted at $j$.
This produces a deterministic and consistent context set for the corrector to use.

\noindent\textbf{SLM-as-a-corrector.}
Given the context set $\mathcal{N}(j)$ for each target column $j$, we formulate cell repair as a sequence generation problem.
For every row $i$ and target $j$, we construct a textual input sequence $\mathbf{x}^{(j)}_i$ as follows:
$$
\mathbf{x}^{(j)}_i = \mathrm{Serialize} \Big(\{(a_k, x_{ik}) : k \in \mathcal{N}(j)\},\, a_j\Big),
$$
where $\mathrm{Serialize}(\cdot)$ maps the selected column-value pair ($a_k$, $x_{ik}$) and the target column ($a_j$) into a token sequence that is shared across all rows~\cite{hegselmann2023tabllm,dinh2022lift}.

We use a pretrained SLM $g_\theta$ as the corrector and fine-tune it by maximizing the conditional log-likelihood over cells not flagged as erroneous by the detector:
$$
\ell(\theta)= \sum_{(i,j)\notin \mathcal{E}} \log p_\theta\!\big(x_{ij}\mid \mathbf{x}^{(j)}_i\big).
$$
At inference time, for each detected erroneous cell $(i,j)\in\mathcal{E}$:
$$
\hat{z}_{ij} = \arg\max_{z}\, p_\theta(z\mid \mathbf{x}^{(j)}_i),
$$
and write $\hat{z}_{ij}$ back to obtain the repaired table.

This design delegates the expensive task of global structure discovery to the instructor, while the corrector provides efficient cell-level repair.

\subsection{\Frameworkname{}+ method}
\label{subsec:em}
\noindent \textbf{Motivation.}
\Frameworkname{} selects a relevant context set for each target column, thereby reducing irrelevant information.
However, column-level relevance alone cannot guarantee cell-level reliability, as the selected columns can still contain erroneous values from the dirty table.
Using such low-quality values as context can mislead the corrector and consequently produce biased repairs.
Addressing this issue requires going beyond selecting relevant context to selecting high-quality context~\cite{jiang2018mentornet}.
We therefore extend \Frameworkname{} to \Frameworkname{}+, which adopts an EM-style procedure~\cite{apd1979em, Neal1998emview, AP1977EM} that alternately refines the repaired table and updates the corrector.

Let $\theta^{(t)}$ denote the parameters of the corrector $g_\theta$ at iteration $t$, and let $\mathbf{D}_r^{(t)}$ denote the repaired table at the beginning of iteration $t$, with $\mathbf{D}_r^{(0)}=\mathbf{D}_e$.
For each row $i$ and target column $j$, we construct the input from current repaired table $D_r^{(t)}$:
$$
x_i^{(j,t)} = \mathrm{Serialize}\Big( \{(a_k,\mathbf{D}_r^{(t)}[i,k]) : k\in\mathcal{N}(j)\},a_j\Big).
$$

\noindent \textbf{M-step: repaired context update.} 
At iteration $t$, the current corrector $g_{\theta^{(t)}}$ induces a predictive distribution $p_{\theta^{(t)}}(z\mid x_i^{(j,t)})$ over candidate repairs. 
For each detected erroneous cell $(i,j)\in\mathcal{E}$, the corrector generates a repair by maximum a posteriori as follows:
$$
\hat z^{(t)}_{ij} = \arg\max_{z} \; p_{\theta^{(t)}}(z\mid x_i^{(j,t)}).
$$
We then construct the next repaired table $\mathbf{D}_r^{(t+1)}$ by replacing each detected erroneous cell with $\hat z^{(t)}_{ij}$.
Since subsequent inputs are constructed from $\mathbf{D}_r^{(t+1)}$, this update turns generated repairs into an improved context for later repair steps.

\noindent \textbf{E-step: corrector update.} 
The corrector is then updated using the repaired context. Specifically, the parameters are updated by maximizing the conditional log-likelihood:
$$
\theta^{(t+1)} = \arg\max_{\theta}\sum_{(i,j)\notin \mathcal{E}} \log p_{\theta}(\hat z^{(t)}_{ij}\mid x_i^{(j,t+1)}).
$$
The update is performed using the repaired table $D_r^{(t+1)}$ as context, thereby adapting the corrector to progressively improved repair contexts.

\noindent\textbf{Effect of iterative refinement.} 
By alternating the above two update steps, \Frameworkname{}+ progressively refines the repaired table and the corrector parameters over successive iterations. 
Early iterations may still be affected by contaminated context, while later iterations benefit from an increasingly reliable repaired table and an updated corrector.
This mutually reinforcing procedure directly addresses context contamination by making the selected context increasingly reliable across iterations. 
We further provide a theoretical monotonicity guarantee for the corresponding conditional log-likelihood.

\begin{theorem}[Monotonicity of Likelihood]
\label{thm:em-monotone}
Let $\theta^{(t)}$ be the corrector parameter vector generated by the EM-style procedure at iteration $t$, and let $\ell(\theta)$ denote the corresponding conditional log-likelihood.
Then, after each EM update, $\ell(\theta)$ is non-decreasing across iterations:
\[
\ell\big(\theta^{(t+1)}, z^{(t+1)}\big)
\;\ge\;
\ell\big(\theta^{(t)}, z^{(t)}\big)
\qquad \text{for all } t\ge 0.
\]
\end{theorem}
\noindent A proof sketch is provided in Section~\ref{thm:em-monotone}. 

\subsection{\Frameworkname++ method}
\label{subsec:cl}

\noindent \textbf{Motivation.}
The EM-style procedure in \Frameworkname{}+ improves the context quality and updates the corrector.
However, the corrector still relies on stochastic optimization methods during training and probabilistic decoding during inference, both of which introduce uncertainty into the generated repairs.
Such uncertainty can accumulate throughout the EM-style procedure, as generated repairs are repeatedly reused to construct contexts for subsequent iterations~\cite{arazo2020pseudolabel, xie2020selftraining}. 
Consequently, without appropriate regulation, unreliable repaired values increasingly bias later corrector updates and repair predictions~\cite{nat2013noisylabel}.
Addressing this issue requires identifying uncertain repairs and limiting their contribution to corrector updates.
We therefore propose \Frameworkname{}++, which introduces confidence-based weighting to filter low-confidence repair candidates and down-weight uncertain repaired rows.
Specifically, confidence-based weighting addresses model uncertainty through three coordinated operations at the cell, column, and row levels:

\noindent\textbf{(1) Cell-level confidence.}
Confidence-based weighting begins by quantifying the confidence of each generated repair at the cell level.
Inspired by confident learning~\cite{northcutt2021CL}, we estimate the confidence by treating predicted repair values as classes.
However, unlike the original confident learning formulation, treating each repair as a distinct class would result in a sparse class space, making confidence estimation unstable.
We therefore use the predictive distribution at the first decoding position to estimate the cell-level confidence~\cite{hendrycks2018baselinedetect}.

Let $\mathbf{Voc}$ denote the corrector vocabulary.
At iteration $t$, for each detected erroneous cell $(i,j)\in\mathcal{E}$, the corrector produces a logit vector over $\mathbf{Voc}$ at the first decoding position. 
We extract the top-$K$ first-token candidates:
$$
A^{(t)}_{ij} = \{ a^{(t)}_{ij(1)},\ldots,a^{(t)}_{ij(K)} \},
$$
together with their logits $\mathrm{logit}^{(t)}_{ij(1)},\ldots,\mathrm{logit}^{(t)}_{ij(K)}$, where $a^{(t)}_{ij(r)}$ denotes the token with the $r$-th largest logit.
We then normalize these logits via a temperature-scaled softmax:
$$
q^{(t)}_{ij(r)}
= \frac{\exp(\mathrm{logit}^{(t)}_{ij(r)}/\tau)}{\sum_{s=1}^{K}\exp(\mathrm{logit}^{(t)}_{ij(s)}/\tau)},
\qquad
\sum_{r=1}^{K} q^{(t)}_{ij(r)}=1,
$$
where $\tau>0$ is the temperature parameter that controls the sharpness of the confidence distribution~\cite{guo2017temperature}. 

\noindent \textbf{(2) Column-level calibration.}
We then calibrate the cell-level confidence using column-level information to reduce confidence bias.
For a target column $j$ and token $v$, let
$$
\mathcal{I}^{(t)}_{j}(v) = \left\{\, (i,r)\ \middle|\ (i,j)\in\mathcal{E},\ a^{(t)}_{ij(r)}=v ,r\in[1,K]\right\}
$$
be the set of positions whose first token is $v$.
For a token $v$ with $\mathcal{I}^{(t)}_j(v)\neq\emptyset$, we define the class-specific threshold as:
$$
T^{(t)}_{j}(v)
= \frac{1}{|\mathcal{I}^{(t)}_{j}(v)|}\sum_{(i,r)\in\mathcal{I}^{(t)}_{j}(v)} q^{(t)}_{ij(r)}.
$$

Given the set of thresholds $\{T^{(t)}_{j}(v)\}_v$, for each detected erroneous cell $(i,j)\in\mathcal{E}$, we construct the corresponding accepted candidate index set as:
$$
\mathcal{R}^{(t)}_{ij}=\{r\in[1,K]\mid q^{(t)}_{ij(r)}\ge T^{(t)}_j(a^{(t)}_{ij(r)})\}.
$$
We then select the index of the most confident accepted token:
$$
\hat r_{ij}^{(t)} = \arg\max_{r\in \mathcal{R}^{(t)}_{ij}}q^{(t)}_{ij(r)},
$$
and use the corresponding first token to generate the remaining repair. 
If no candidate exceeds the threshold, we set $\hat r^{(t)}_{ij}=\varnothing$.
This threshold calibrates confidence within each target column.

\noindent\textbf{(3) Row-level aggregation.} 
We next aggregate the calibrated confidence into row-level weights for corrector updates. 
For cells not flagged as erroneous, we set the confidence to 1.
For accepted repairs, we use the probability of the selected first token as the confidence.
If no candidate passes the threshold, we set the confidence to $0$.
$$
c^{(t)}_{ij}=
\begin{cases}
1, & (i,j)\notin\mathcal{E},\\[4pt]
q^{(t)}_{ij(\hat r^{(t)}_{ij})}, & (i,j)\in\mathcal{E}\ \land\ \hat r^{(t)}_{ij}\neq\varnothing,\\[4pt]
0, & (i,j)\in\mathcal{E}\ \land\ \hat r^{(t)}_{ij}=\varnothing.
\end{cases}
$$

The calibrated confidence is then aggregated into a row-level reliability weight, denoted by $\omega^{(t)}_i=\frac{1}{d}\sum_{j=1}^d c^{(t)}_{ij}$, and let $\omega_i^{\ast(t)}$ be its normalized form.
This aggregation assigns smaller weights to rows containing low-confidence or rejected repairs, thereby limiting their influence on subsequent updates.

After the confidence-based weighting, we refine the E-step to update the corrector more reliably:
$$
\theta^{(t+1)}
=
\arg\max_{\theta}\;
\sum_{(i,j)\notin\mathcal{E}}
\omega^{\ast(t)}_i\,
\log p_{\theta}\!\left(\hat z^{(t)}_{ij}\mid x^{(j,t)}_i\right).
$$
By integrating confidence-based weighting in this way, \Frameworkname{}++ emphasizes training with more reliably generated repairs while down-weighting those affected by uncertainty.
This reduces the propagation of model uncertainty during the EM-style procedure.
We also state the theorem on the effectiveness of the confidence-based weighting.
\begin{theorem}[Effectiveness of confidence-based weighting]
\label{thm:cl_grad}
Let the weighted gradient be $\hat g_{\omega}$, the unweighted gradient be $\hat g_{u}$, and the true gradient be $g^\ast$. We have
$$
\mathbb{E}\|\hat g_{\omega}(\theta)-g^\ast(\theta)\|^2 \leq \mathbb{E}\|\hat g_u(\theta)-g^\ast(\theta)\|^2.
$$
This indicates that confidence-based weighting reduces the gradient variance induced by uncertainty.
\end{theorem}
\noindent A proof sketch is provided in Section~\ref{subsec:cl_sktech}. 

\subsection{Overall algorithm}
\label{subsec:algorithm}
We summarize the complete workflow of \Frameworkname{}++ in Algorithm~\ref{alg:overall}. 
Given a dirty table $\mathbf{D}_e$ and a detected error set $\mathcal{E}$, \Frameworkname{}++ first invokes the LLM as an instructor to compute a score matrix $\mathbf{W}$ and constructs a directed weighted graph $\mathcal{G}$ based on $\mathbf{W}$.
It then derives a context set $\mathcal{N}(j)$ for each target column $j$ (Lines 2--5).

After the LLM-as-an-instructor stage, \Frameworkname{}++ performs the EM-style procedure in Lines 6--23.
Each iteration first executes the M-step, where the corrector generates a repaired value for each detected erroneous cell under the selected context to update the repaired table $\mathbf{D}_r^{(t+1)}$ (Lines 7--11).
It then executes the E-step, which incorporates confidence-based weighting to update the corrector (Lines 12--22).
Specifically, the algorithm first estimates cell-level confidence over the top-$K$ first-token candidates (Lines~12--14).
It then performs column-level calibration by computing confidence thresholds and selecting the accepted first token (Lines 15--17).
Next, the calibrated cell-level confidence is aggregated into row-level weights (Lines 18--21).
Finally, these confidence-based weights are incorporated into the corrector update (Line 22).
The algorithm terminates when the confidence-based weights converge or the maximum number of iterations is reached, and returns the final repaired table $\mathbf{D}_r^{(t)}$.

\begin{algorithm}[t]
\caption{Overall Workflow of \Frameworkname{}++}
\LinesNumbered
\DontPrintSemicolon
\label{alg:overall}
\small
\KwIn {dirty table $\mathbf{D}_e$, detected error set $\mathcal{E}$, LLM $\mathcal{M}$, SLM $g_{\theta^{(0)}}$, maximum iterations $T$, hop budget $h$, top-$K$ size $k$, context threshold $\epsilon$, convergence tolerance $\eta$, temperature $\tau$. \\}
\KwOut{repaired table $\mathbf{D}_r$.} 
$\mathbf{D}_r^{(0)} \leftarrow \mathbf{D}_e$\\
\tcp*[l]{\color{blue}\textbf{Stage I: LLM as an instructor}}
$\mathbf{W} \leftarrow \textsc{QueryLLM}(\mathcal{M},\, \textsc{Sample}(\mathbf{D}_e))$ \\
$\mathcal{G} \leftarrow \textsc{WeightGraph}(\mathbf{W})$\\
\For{$j \in \mathcal{C}$}{
    $\mathcal{N}(j) \leftarrow \textsc{KHopTreeNeighbour}(\mathcal{G},h,j, \epsilon)$ \\
}
\tcp*[l]{\color{blue}\textbf{Stage II: EM-style procedure}}
\While{$t\le T$ \textbf{and} $ \Delta_{\omega}\geq \eta$}{
    \tcp*[l]{\color{blue}\textbf{M-step: repaired context update}}
    $\mathbf{D}_r^{(t+1)}\leftarrow \mathbf{D}_r^{(t)}$\\
    \For{$(i,j)\in\mathcal{E}$}{
        $\mathbf{x}_{i}^{(j,t)}\leftarrow\textsc{Serialize}(\mathbf{D}_r^{(t)},i,\mathcal{N}(j))$\\
        $z_{ij}^{(t)}\leftarrow \arg\max_zp_{\theta^{(t)}}(z\mid x_i^{(j,t)})$\\
        $\mathbf{D}_r^{(t+1)}[i,j]\leftarrow z_{ij}^{(t)}$
    }
    \tcp*[l]{\color{blue}\textbf{E-step: corrector update}}
    \tcp*[l]{\color{blue}\textbf{Cell-level confidence}}
    \For{$(i,j)\in\mathcal{E}$}{
        $(A_{ij}^{(t)},\mathrm{logit}_{ij}^{(t)})\leftarrow \textsc{TopKToken}(g_\theta^{(t)},\mathbf{x}_i^{(j,t)},k)$\\
        $q_{ij}^{(t)}\leftarrow \textsc{Softmax}(\mathrm{logit}_{ij}^{(t)},\tau)$
    }
    \tcp*[l]{\color{blue}\textbf{Column-level calibration}}
    \For{$j\in\mathcal{C}$}{
        $T_j^{(t)}(\cdot)\leftarrow \textsc{TokThreshold}(A_{ij}^{(t)},q_{ij}^{(t)}:(i,j)\in\mathcal{E})$\\
        $\hat r_{ij}^{(t)} \leftarrow \textsc{SelectTok}\big(A_{ij}^{(t)}, q_{ij}^{(t)}, T_j^{(t)}(\cdot)\big)$
    }
    \tcp*[l]{\color{blue}\textbf{Row-level aggregation}}
    \For{all $(i,j)$}{
        $c_{ij}^{(t)}\leftarrow \textsc{Confidence}(q_{ij}^{(t)},\hat r_{ij}^{(t)})$
    }
    $\omega^{(t)}\leftarrow \textsc{Aggregation}(c^{(t)})$\\
    $\omega^{\ast(t)}\leftarrow \textsc{Normalize}(\omega^{(t)})$\\
    $\theta^{(t+1)}\leftarrow \textsc{WeightUpdate}(\theta^{(t)},\omega^{\ast(t)})$\\
    $t\leftarrow t+1$
}

\textbf{return} $\mathbf{D}_r^{(t)}$
\end{algorithm}

\section{ANALYSIS}
 \subsection{Time complexity analysis}
 \label{subsec:time-analysis}
We analyze the time complexity of \Frameworkname{}++ by separating the instructor stage and the corrector stage.
Since the instructor is invoked once on a compact sketch, we exclude external LLM inference from the asymptotic complexity and focus on the computation performed by the corrector.

After obtaining the score matrix $\mathbf{W}$, for each target column $j\in\{1,\dots,d\}$, we construct an $h$-hop influence tree and select the context set via the threshold $\epsilon$.
A straightforward dynamic programming implementation costs $O(hd^2)$ per target column, yielding a total cost of $O(hd^3)$ for graph construction.

For the EM-style procedure, let $|\mathcal{E}|$ be the number of detected erroneous cells, and let $C_{\mathrm{SLM}}$ denote the average per-cell cost of one corrector pass on a serialized input. 
At each iteration of the EM-style procedure, the corrector generates repairs and confidence for cells in $\mathcal{E}$, and is then fine-tuned using the confidence-based weights.
Therefore, for $e$ fine-tuning epochs over $T$ iterations, the dominant cost is $O(T\cdot \textit{e}\cdot |\mathcal{E}|\cdot C_{\mathrm{SLM}})$.
The additional cost of computing class-specific thresholds and confidence-based weights is linear in the number of candidates and detected erroneous cells, and is dominated by corrector inference and training in practice.

\subsection{Effectiveness of the EM-style procedure}
\label{subsec:em-effectiveness}

We now justify the monotonicity claim in Theorem~\ref{thm:em-monotone} under the idealized setting where the repaired context update and the corrector update are solved exactly. 
Let $\{(\theta^{(t)}, \hat{\mathbf z}^{(t)})\}_{t\ge 0}$ be the sequence generated by the updates in Section~\ref{subsec:em}, where $\hat{\mathbf z}^{(t)} = \{\hat z_{ij}^{(t)}\}_{(i,j)\in\mathcal{E}}$.
For a fixed set of serialized repaired inputs, define the conditional log-likelihood as
$$
\ell(\theta,\hat{\mathbf z})=\sum_{(i,j)\in\mathcal{E}}
\log p_\theta(\hat z_{ij}\mid x_i^{(j)}).
$$
\noindent\textbf{Proof sketch of Theorem~\ref{thm:em-monotone}.}
The argument follows from the properties of the alternating optimization procedure.

\textit{Repaired context update.}
Given $\theta^{(t)}$, the repaired context update selects the most likely repair value for each detected erroneous cell. 
Therefore, under the exact update assumption,
$$
\ell(\theta^{(t)},\hat{\mathbf z}^{(t+1)})\ge\ell(\theta^{(t)},\hat{\mathbf z}^{(t)}).
$$

\textit{Corrector update.}
Given $\hat{\mathbf z}^{(t+1)}$, the corrector update optimizes the same objective with respect to $\theta$. Hence,
$$
\ell(\theta^{(t+1)},\hat{\mathbf z}^{(t+1)})\ge\ell(\theta^{(t)},\hat{\mathbf z}^{(t+1)}).
$$

Combining the two inequalities yields
$$
\ell(\theta^{(t+1)},\hat{\mathbf z}^{(t+1)})
\ge
\ell(\theta^{(t)},\hat{\mathbf z}^{(t+1)})
\ge
\ell(\theta^{(t)},\hat{\mathbf z}^{(t)}),
$$
which proves the monotonic non-decrease of the conditional log-likelihood under the idealized EM-style assumptions.

For \Frameworkname{}++, the confidence-based weights are computed before the weighted corrector update and then fixed during that update.
Therefore, the same alternating-optimization argument applies to the corresponding weighted objective (the confidence-based weighting).
The complete proof is provided in our supplementary repository~\cite{lasrepair_proofs}.

\subsection{Effectiveness of confidence-based weighting}
\label{subsec:cl-theory}

We now theoretically analyze the effect of the confidence-based weighting in \Frameworkname{}++. 
As introduced in Section~\ref{subsec:cl}, \Frameworkname{}++ aggregates cell-level confidence into confidence-based weights and uses the normalized weights $\omega_i^\ast$ to reduce the gradient variance induced by uncertainty.
We now justify that claim for Theorem~\ref{thm:cl_grad}. Consider one fixed iteration and omit the superscript $t$.


\medskip
\noindent\textbf{Proof sketch of Theorem~\ref{thm:cl_grad}.}
\label{subsec:cl_sktech}
For each row $i$, define the loss
$$
\hat \ell_i(\theta) = \sum_{(i,j)\in\mathcal{E}}-\log p_\theta(\hat z_{ij}\mid x_i^{(j)}),
$$
and the ground-truth loss
$$
\ell^\ast_i(\theta) = \sum_{(i,j)\in\mathcal{E}}-\log p_\theta(z_{ij}\mid x_i^{(j)}).
$$
Let $\hat g_i(\theta)=\nabla_\theta \hat \ell_i(\theta)$ and $g_i^\ast(\theta)=\nabla_\theta \ell_i^\ast(\theta)$ denote the gradients computed from generated repairs and ground-truth values, respectively. 
Assume that
$$
\hat g_i(\theta)=g_i^\ast(\theta)+\epsilon_i,
\qquad
\epsilon_i\sim N(0,\sigma^2/\omega_i^\ast),
$$
where $\{\epsilon_i\}_{i=1}^n$ are conditionally independent. 
For the gradients in the Theorem~\ref{thm:cl_grad}, we have
$\hat g_{\omega}(\theta)=\sum_{i=1}^n \omega_i^\ast \hat g_i(\theta)$,
$\hat g_u(\theta)=\frac1n\sum_{i=1}^n \hat g_i(\theta)$ and
$g^\ast(\theta)=\frac1n\sum_{i=1}^n g_i^\ast(\theta)$.

Under this setting, rows with higher confidence induce lower gradient variance. 
Therefore, the weighted estimator assigns less weight to uncertain rows than does uniform averaging.
A direct variance comparison, together with the Cauchy--Schwarz inequality, yields
$$
\mathbb{E}\|\hat g_{\omega}(\theta)-g^\ast(\theta)\|^2
\le
\mathbb{E}\|\hat g_u(\theta)-g^\ast(\theta)\|^2.
$$
This proves the claim in Theorem~\ref{thm:cl_grad}. The full proof is provided in our supplementary repository~\cite{lasrepair_proofs}.
\section{Experiments}
\subsection{Dataset description}
We evaluate \Frameworkname{}++ on 7 benchmark datasets widely used in prior data-cleaning studies~\cite{wei2024survey, gao2019mlnclean, mohammad2020baran}. 
Table~\ref{tab:datasets} summarizes their statistics, including dataset size, error types, cell-level error rates, and row-level error rates. 
The error types include missing value (MV), typo (T), violated attribute dependency (VAD), and formatting issue (FI). 

\begin{table}[t]
\centering
\caption{Statistics of the datasets. R-rate and C-rate are Row-level error rate and Cell-level error rate, respectively.}
\label{tab:datasets}
\vspace{-0.2cm}
\scalebox{0.8}{
\begin{tabular}{|c|c|c|c|c|c|}
\hline
\textbf{Datasets}
  & \textbf{\#Tuples}
  & \textbf{\#Attrs}
  & \textbf{R-rate, \%}
  & \textbf{C-rate, \%}
  & \textbf{Error type} \\
\hline\hline
Hospital & 1,000 & 20 & 40.70 & 2.67 & T, VAD      \\
Flight   & 2,376 & 7  & 80.13 & 34.51 & MV, FI, VAD \\
Beers    & 2,410 & 11 & 100.00 & 13.93 & MV, FI, VAD \\
Walmart   & 4,653 & 5  & 62.91 & 17.92 & SY \\
Tax\_20k  & 20,000  & 15  & 4.76 & 0.32 & T, FI, VAD \\
Shuttle   & 43,500 & 10  & 73.55 & 12.44 & SY \\
Tax\_200k & 200,000 & 15  & 1.46 & 0.10 & T, FI, VAD \\
\hline
\end{tabular}
}
\end{table}

\noindent\textbf{Synthetic error generation.} For datasets marked as SY, only a clean version is available. 
Therefore, we follow the prior work~\cite{errorinje2015aro} to inject errors and construct the dirty version.
The injection process follows the Missing at Random (MAR) assumption~\cite{rubin1976mar}, where corruption depends on observed information but not the unknown clean value.
Specifically, we first sample cells according to a target error rate.
For each selected cell, we sample an error type based on the column data type and apply the corresponding operator.
For MV, we replace the value with an empty token.
For T, we apply $1$--$k$ character-level edits to strings, where $k\in\{1,2,3\}$ and $k\le |x_{ij}|/2$. Each edit is uniformly sampled from insertion, deletion, substitution, and transposition.
For numerical columns, VAD adds zero-mean noise scaled by column statistics, while FI applies format transformations such as integer-to-float conversion.

\subsection{Experimental setup}  
\noindent\textbf{Baselines.} We compare \Frameworkname{}++ with 8 representative baselines, covering both non-generative and generative data repair methods.
For compact presentation, we group the conventional and discriminative baselines as non-generative methods in the experimental tables:
1) Baran~\cite{mohammad2020baran}, a data-driven repair framework that generates candidate repairs and selects corrections using learned classifiers;
2) HoloClean~\cite{theo2017holoclean}, which combines integrity constraints, external signals, and statistical inference to identify likely correct values;
3) BoostClean~\cite{sanjay2017boostclean}, which selects cleaning operations to improve downstream data quality;
4) Unified~\cite{fei2011unified}, which performs tolerant repair under functional dependencies using a minimum description length principle; 
5) BigDansing~\cite{zuhair2015bigdansing}, which accelerates rule-based data repair with two optimization strategies;
6) Holistic~\cite{xu2013holistic}, which models equivalent classes and conflicts via a hypergraph for data repair.
The generative baselines include:
7) Jellyfish~\cite{zhang2024jellyfish}, an LLM-based model for data preprocessing and repair;
8) GIDCL~\cite{yan2024gidcl}, a graph-enhanced LLM-based data cleaning framework.

\noindent\textbf{Metrics.} 
We evaluate repair quality using three complementary metrics: F1-score, Error Drop Rate (EDR), and Record Distance Reduction Rate (RDRR).
F1-score~\cite{sasaki2007truth} is the primary metric and is computed from exact match repair outcomes following prior work.
To further characterize repair performance, we additionally report EDR~\cite{wei2024survey} and RDRR, which capture complementary aspects of repair quality.

EDR measures the relative reduction in the number of erroneous cells after repair:
$$
\text{EDR}=\frac{\text{OEC}-\text{REC}}{\text{OEC}},
$$
where $\text{OEC}$ and $\text{REC}$ denote the \textit{original error count} and \textit{remaining error count}, respectively. 
A larger EDR indicates that more original errors are removed.

To evaluate partial repairs that do not exactly match the ground truth, we further introduce RDRR based on the Levenshtein edit distance $d_{\mathrm{lev}}$~\cite{levenshtein1966binary}.
Let $g_i$ denote the ground-truth value and $r_i$ denote the corresponding value in either the dirty or repaired dataset. 
Let $N$ denote the number of evaluated cells. We first compute the average normalized edit distance:
$$
\text{AED}=\frac{1}{N}\sum_{i=1}^N \frac{d_{\mathrm{lev}}(g_i, r_i)}{\max(|g_i|, |r_i|)},
$$
RDRR is then defined as
$$
\text{RDRR} = \frac{\text{AED}_d - \text{AED}_r}{\text{AED}_d},
$$
where the subscripts $d$ and $r$ denote the dirty and repaired datasets, respectively. 
RDRR measures the relative reduction in edit distance achieved by repair. 
It is particularly useful when a method improves a corrupted value substantially but does not recover the exact ground truth. Higher values of all three metrics indicate better performance. \\

\noindent\textbf{Implementation details.} All baselines are executed with their default hyperparameters when available. 
We use \textit{Jellyfish-7B} in the Jellyfish and GIDCL experiments. 
For \Frameworkname{}++, we use \textit{T5-large}~\cite{raffel2023t5}, with 770M parameters, as the SLM corrector, and GPT-5~\cite{singh2026gpt5} as the instructor. 
Unless otherwise specified, we train the corrector for 3 epochs, run at most 10 EM-style iterations, and set the softmax temperature to $\tau=1$. 
For large-scale datasets, we sample 2,000 records to fine-tune the corrector while evaluating repair quality on the target cells. 
All experiments are conducted on a server with an Intel Xeon Silver 4313 CPU and an NVIDIA RTX A5000 GPU.

\newcolumntype{P}[1]{>{\centering\arraybackslash}p{#1}}
\newcommand{\mcolw}{0.88cm} 

\begin{table*}[t]
\caption{Overall repair performance on 7 benchmark datasets. We compare non-generative and generative repair methods using F1, EDR, and RDRR, where higher values indicate better quality. The best result is highlighted in bold and underlined, and the second-best result is highlighted in bold. The last block reports average improvement of \Frameworkname{}++ over each baseline. OOM and OOT denote out-of-memory and out-of-time failures.}

\vspace{-0.3cm}
\label{tab:f1score}
\centering
\small
\setlength{\tabcolsep}{2pt}
\renewcommand{\arraystretch}{0.98}
\begin{adjustbox}{max width=\textwidth}
\begin{tabular}{|P{2.1cm}|P{1.0cm}|*{6}{P{1.4cm}|}*{3}{P{1.4cm}|}}
\hline
\multirow{2}{*}{\textbf{Datasets}} &
\multirow{2}{*}{\textbf{Metrics}} &
\multicolumn{6}{c|}{\textbf{Non-generative}} &
\multicolumn{3}{c|}{\textbf{Generative}} \\
\cline{3-11}
& &
\multicolumn{1}{|>{\columncolor{gray!15}}c|}{\textbf{Baran}} &
\multicolumn{1}{>{\columncolor{gray!15}}c|}{\textbf{HoloClean}} &
\multicolumn{1}{>{\columncolor{gray!15}}c|}{\textbf{BoostClean}} &
\multicolumn{1}{>{\columncolor{gray!15}}c|}{\textbf{Unified}} &
\multicolumn{1}{>{\columncolor{gray!15}}c|}{\textbf{BigDansing}} &
\multicolumn{1}{>{\columncolor{gray!15}}c|}{\textbf{Holistic}} &
\multicolumn{1}{>{\columncolor{gray!15}}c|}{\textbf{Jellyfish}} &
\multicolumn{1}{>{\columncolor{gray!15}}c|}{\textbf{GIDCL}} &
\multicolumn{1}{>{\columncolor{gray!15}}c|}{\textbf{\Frameworkname++}} 
\\
\hline

\multirow{3}{*}{Hospital}
 & F1   & 0.5844 & 0.6262 & 0.3312 & 0.7825 & 0.6210 & 0.6101 & 0.8871 & \second{0.8925} &  \best{0.9632} \\
 & EDR  & 0.4165 & 0.4538 & 0.3132 & 0.7612 & -0.0716 & -0.0187 & 0.8730 & \second{0.8822} &   \best{0.9871} \\
 & RDRR & 0.3548 & 0.4432 & 0.1037 & 0.5960 & 0.0245 & 0.0449 & 0.8014 & \second{0.8256} &  \best{0.8647} \\
\hline

\multirow{3}{*}{Flights}
 & F1   & \second{0.6369} & 0.4734 & 0.0620 & 0.5075 & 0.1090 & 0.2313 & 0.6324 & 0.5903 &  \best{0.8798} \\
 & EDR  & 0.4913 & -0.1339 & -0.0028 & 0.0541 & -0.3543 & -0.1423 & \second{0.5919} & 0.4872 &  \best{0.9072} \\
 & RDRR & 0.5315 & -0.2341 & -0.0039 & 0.0411 & -1.0856 & -0.7324 & \second{0.5830} & 0.3156 &  \best{0.9683} \\
\hline

\multirow{3}{*}{Beers}
 & F1   & 0.7576 & 0.0467 & 0.0109 & 0.0324 & 0.1157 & 0.0892 & \second{0.7832} & 0.4332 &  \best{0.8894} \\
 & EDR  & 0.7082 & -3.6209 & -0.7192 & -0.2442 & -0.0108 & -0.0093 & \second{0.7264} & -0.1335 &  \best{0.8838} \\
 & RDRR & \second{0.7842} & -5.7141 & -0.9923 & -0.3121 & -0.1192 & -0.1372 & 0.7197 & -0.2290 & \best{0.9253} \\
\hline

\multirow{3}{*}{Walmart}
 & F1   & 0.2460 & \second{0.5315} & 0.3167 & 0.3528 & 0.3797 & 0.2912 & 0.4712 & 0.4372 &  \best{0.7409} \\
 & EDR  & 0.0368 & 0.0276 & 0.1049 & 0.3006 & -0.2334 & -0.3876 & \second{0.3122} & -0.0897 &  \best{0.5317} \\
 & RDRR & 0.0113 & 0.0198 & 0.1953 & \second{0.4414} & -0.5298 & -0.4893 & 0.3464 & -0.2315 & \best{0.5417} \\
\hline

\multirow{3}{*}{Tax\_20k}
 & F1   & 0.3512 & 0.0000 & 0.0000 & \textsc{oot} & 0.2413 & 0.2413 & 0.5774 & \second{0.5923} & \best{0.6944} \\
 & EDR  & 0.2490 & -0.0198 & 0.0000 & -- & -0.9127 & -0.9127 & \second{0.5008} & 0.4863 &  \best{0.6793} \\
 & RDRR & 0.1738 & -0.6295 & -0.3957 & -- & -0.3531 & -0.3531 & 0.6022 & \second{0.6139} &  \best{0.7129} \\
\hline

\multirow{3}{*}{Shuttle}
 & F1   & 0.7775 & \textsc{oom} & 0.0000 & 0.0000 & \textsc{oot} & \textsc{oot} & \second{0.8359} & 0.7209 & \best{0.8898} \\
 & EDR  & 0.6931 & -- & -0.4932 & -0.3315 & -- & -- & \best{0.8142} & 0.6559 &  \second{0.7652} \\
 & RDRR & 0.7032 & -- & -0.6587 & -0.5114 & -- & -- & \second{0.8462} & 0.7621 & \best{0.9708} \\
\hline

\multirow{3}{*}{Tax\_200k}
 & F1   & 0.3276 & \textsc{oom} & 0.0000 & \textsc{oot} & \textsc{oot} & \textsc{oot} & 0.5227 & \second{0.5450} &   \best{0.6332} \\
 & EDR  & 0.1814 & -- & 0.0000 & -- & -- & -- & 0.4432 & \second{0.4573} & \best{0.6372} \\
 & RDRR & 0.2321 & -- & -0.4498 & -- & -- & -- & 0.5486 & \second{0.5686} &  \best{0.7567} \\
\hline

\cellcolor{gray!15}\multirow{3}{*}{\shortstack[c]{\textbf{Average}\\\textbf{Improvement}}}
& F1   & 0.2866 & 0.5728 & 0.7095 & 0.5732 & 0.6030 & 0.6035 & 0.1397 & 0.2109 & -- \\
\cellcolor{gray!15}\multirow{-3}{*}{}
& EDR  & 0.3776 & 1.2446 & 0.8880 & 0.6970 & 1.0003 & 0.9843 & 0.1654 & 0.3819 & --   \\
\cellcolor{gray!15}\multirow{-3}{*}{\shortstack[c]{\textbf{Average}\\\textbf{Improvement}}}
& RDRR & 0.4213 & 1.6935 & 1.1345 & 0.7836 & 1.1148 & 1.0582 & 0.1847 & 0.4450 & --   \\
\hline

\end{tabular}
\end{adjustbox}
\end{table*}

\subsection{Effectiveness evaluation}

\noindent\textbf{Exp-1: F1-score, EDR, and RDRR.}
We first evaluate the effectiveness of \Frameworkname{}++ against 8 representative baselines on 7 benchmark datasets.
Table~\ref{tab:f1score} reports the F1-score, EDR, and RDRR of all methods across the benchmark datasets. 
Some non-generative methods fail to finish on large datasets within the time or memory budget. 
These cases are marked as OOT (out of time) or OOM (out of memory).

Overall, \Frameworkname{}++ achieves the best F1-score on all 7 datasets.
Compared with the dataset-wise strongest baseline, \Frameworkname{}++ improves the average F1-score from 0.6882 to 0.8130, yielding a relative gain of 18.1\%.
This demonstrates that the proposed instructor--corrector design, EM-style procedure, and confidence-based weighting jointly improve repair accuracy.
Among the non-generative methods, Baran is generally the strongest baseline.
However, its performance varies substantially across datasets, especially when the candidate generation is insufficient.
Among the generative methods, Jellyfish and GIDCL typically outperform most non-generative baselines, demonstrating the benefits of LLM-based repair.
Even against these strong generative competitors, \Frameworkname{}++ remains consistently superior, achieving an average F1-score of 0.8130. This indicates that directly relying on large generative models is less effective than using an LLM as an instructor and an SLM as a corrector for repair.

The same trend is reflected for EDR and RDRR. 
Compared with the dataset-wise strongest baselines, \Frameworkname{}++ improves the average EDR from 0.6121 to 0.7702 and the average RDRR from 0.6661 to 0.8201.
\Frameworkname{}++ obtains the highest EDR on 6 out of the 7 datasets and the highest RDRR on all datasets. 
This suggests that \Frameworkname{}++ not only exactly corrects more erroneous cells, but also moves incorrect values closer to the ground truth.



\begin{figure}[t]
\centering
\includegraphics[width=\linewidth]{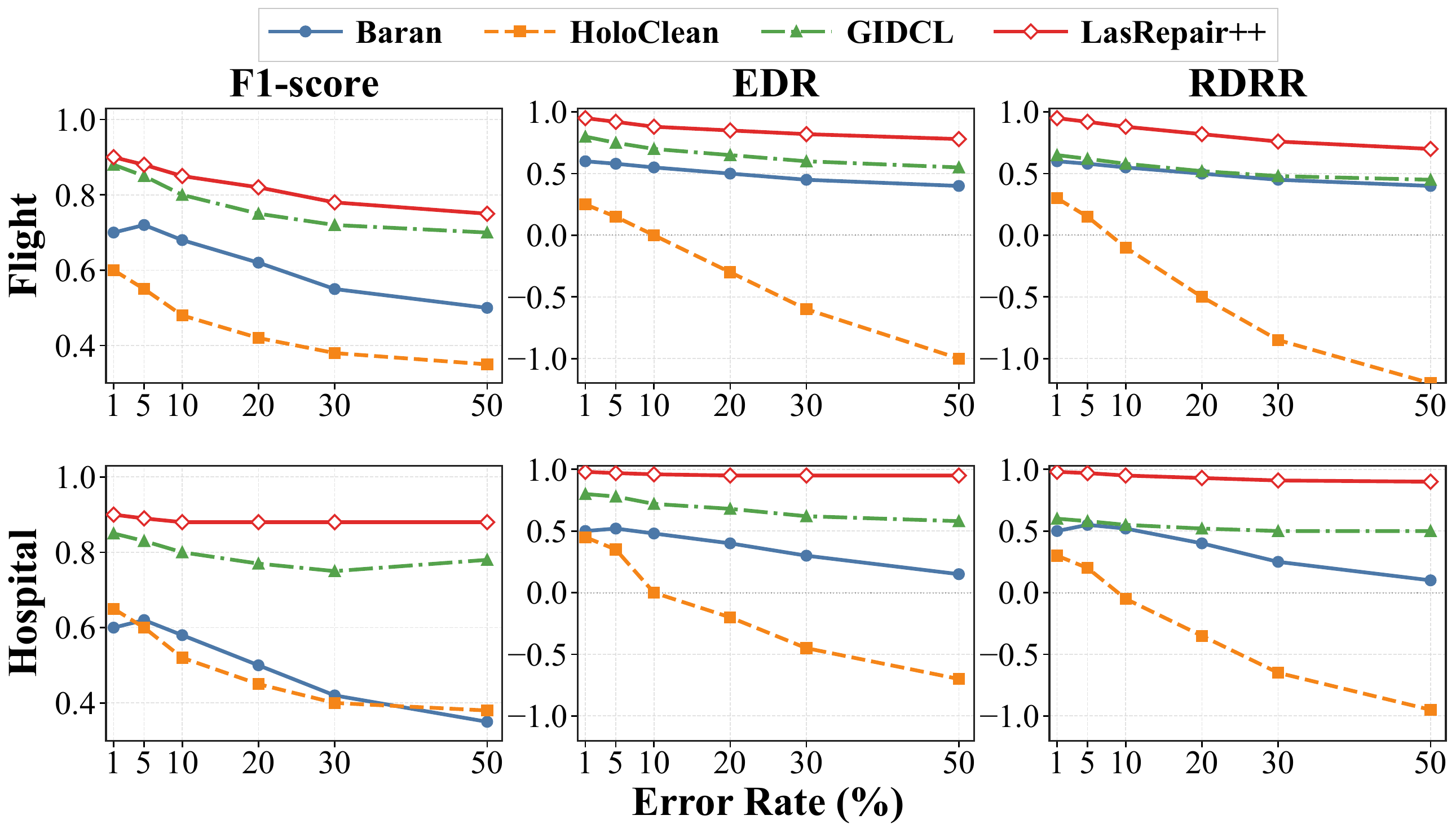}
\caption{Repair performance under varying cell-level error rates on Flight (a--c) and Hospital (d--f), evaluated by F1-score, EDR, and RDRR.}
\label{fig:err_rate}
\end{figure}

\noindent\textbf{Exp-2: Robustness under varying error rates.}
We next evaluate whether \Frameworkname{}++ remains effective when the input dataset becomes increasingly contaminated.
We select Flight and Hospital for this experiment. These two datasets cover all error types considered in our evaluation and are moderately sized, enabling controlled analysis across different error rates.
Specifically, we vary the injected error rate from 1\% to 50\%, and compare \Frameworkname{}++ with representative baselines from both the non-generative and generative categories. 

Figure~\ref{fig:err_rate} reports F1-score, EDR, and RDRR under different error rates.
As the error rate increases, the performance of all methods degrades. 
By comparison, \Frameworkname{}++ degrades more slowly and consistently achieves the best performance across all three metrics on both datasets. 
At the highest error rate of 50\%, \Frameworkname{}++ still outperforms the strongest baseline, GIDCL, by 7.22\% and 19.53\% in F1-score on the Flight and Hospital datasets, respectively.
These results indicate that \Frameworkname{}++ is more robust to severe corruption and remains effective even when the input dataset is highly contaminated.

\begin{figure}[t]
  \centering
  \includegraphics[width=\linewidth, trim={0 5mm 0 0},clip]{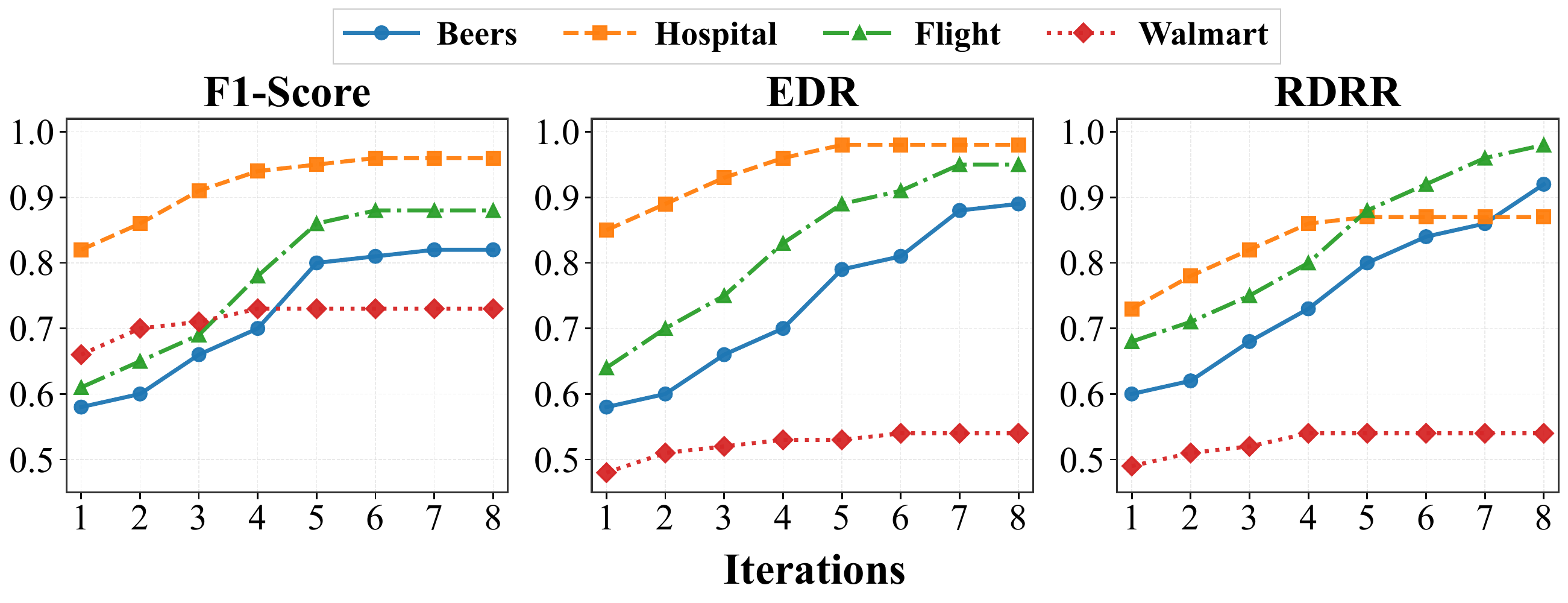}
  \caption{Performance under different numbers of iterations on datasets Beers, Hospital, Flight, and Walmart. We report F1-score, EDR, and RDRR in (a)–(c), respectively.}
  \label{fig:exp3}
\end{figure}

\noindent\textbf{Exp-3: Performance under different numbers of iterations.}
We further study the effect of the EM-style procedure. 
We run \Frameworkname{}++ with different numbers of iterations on Beers, Hospital, Flight, and Walmart, and report F1-score, EDR, and RDRR in Figure~\ref{fig:exp3}.
These datasets exhibit diverse repair difficulties, providing a representative view of how repair quality evolves across iterations.
The results show that repair performance generally improves as the number of iterations increases.
The average F1-score increases from 0.6689 at the first iteration to 0.8662 at the final iteration, yielding a relative improvement of 29.50\%.
Similar trends are observed for EDR and RDRR, which improve by approximately 30.6\% and 31.7\% on average across the 4 datasets.

These results confirm that updating the repaired dataset and fine-tuning the corrector with generated repairs can progressively improve the quality of repairs.
The improvement is sometimes stepwise rather than smooth, because the corrector may learn a recurring repair pattern after several iterations and then correct many similar cells. 
The marginal gain decreases in later iterations, and most improvements are achieved within the first 6--7 rounds.
We therefore set the maximum number of iterations to $T=10$ in the remaining experiments to balance between effectiveness and computational cost.

\subsection{Parameter analysis}
\noindent\textbf{Exp-4: Varying the number of epochs.}
\begin{figure}[t]
  \centering
  \includegraphics[width=\linewidth, trim=0 20 0 0]{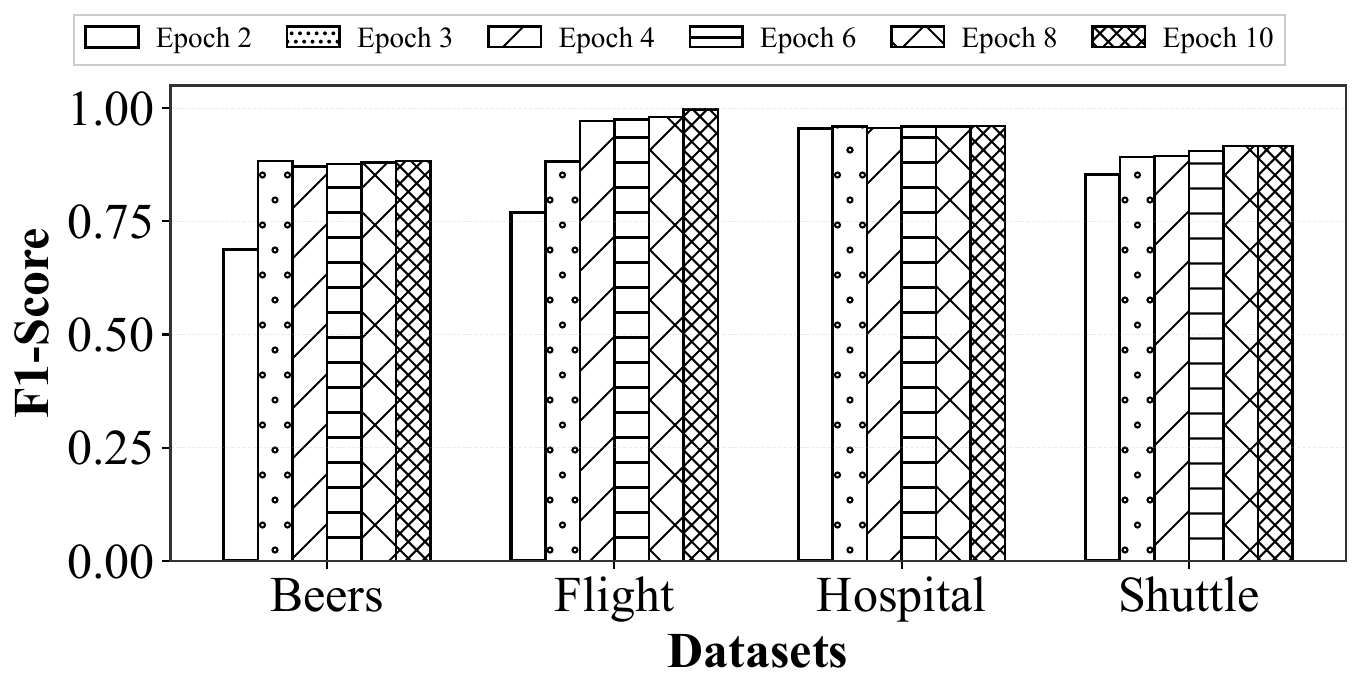}
  \caption{Impact of fine-tuning epochs on repair performance across datasets. F1-score improves noticeably from 2 to 3 epochs, while gains become marginal after about 6 epochs.}
  \label{fig:exp5}
\end{figure}
We study the sensitivity of \Frameworkname{}++ to the fine-tuning budget of the corrector.
We select 4 datasets spanning different data scales, allowing us to examine whether the effect of the training budget remains consistent.
Specifically, we vary the number of fine-tuning epochs over ${2, 3, 4, 6, 8, 10}$ and report the resulting F1-scores in Figure~\ref{fig:exp5}.

Overall, increasing the epoch number from 2 to 3 yields clear improvements across datasets.
This indicates that a very small training budget is insufficient for the corrector to adapt to dataset-specific repair patterns.
Increasing the epoch number from 2 to 3 improves the average F1-score from 0.8169 to 0.9043, while further increasing it to 10 epochs only raises the average F1-score to 0.9397.
This suggests diminishing returns from longer training.
We therefore use 3 epochs as the default setting in the remaining experiments, which provides a good balance between accuracy and computational cost.

\begin{table}[t]
    \centering
    \caption{Effect of the softmax temperature ($\tau$) on repair performance.}
    \begin{adjustbox}{max width=\columnwidth}
    \begin{tabular}{l|ccc|ccc|ccc}
        \toprule
        \multirow{2}{*}{Dataset}
            & \multicolumn{3}{c}{0.5}
            & \multicolumn{3}{c}{1.0}
            & \multicolumn{3}{c}{2.0} \\
        \cmidrule(lr){2-4}\cmidrule(lr){5-7}\cmidrule(lr){8-10}
            & F1 & EDR & RDRR
            & F1 & EDR & RDRR
            & F1 & EDR & RDRR \\
        \midrule
        Beers
            & 0.8721 & 0.8697 & 0.9147
            & 0.8831 & 0.8807 & 0.9262
            & 0.8554 & 0.8531 & 0.8971 \\
        Shuttle
            & 0.8591 & 0.8507 & 0.8913
            & 0.8920 & 0.8833 & 0.9254
            & 0.8701 & 0.8616 & 0.9027 \\
        \bottomrule
    \end{tabular}
    \end{adjustbox}
    \label{tab:exp6}
\end{table}

\noindent\textbf{Exp-5: Varying temperature in confidence-based weighting.} 
We next evaluate the effect of the softmax temperature $\tau$ in the confidence-based weighting.
We vary $\tau$ in $\{0.5, 1.0, 2.0\}$ and report F1-score, EDR, and RDRR on 2 representative datasets with substantially different scales, as summarized in Table~\ref{tab:exp6}. 
Overall, $\tau=1.0$ achieves the best overall performance on both datasets across all 3 metrics, indicating that a moderate temperature yields the most effective confidence-based weighting. 
Compared with $\tau=0.5$ and $\tau=2.0$, $\tau=1.0$ improves the average F1-score by 2.54\% and 2.87\%, respectively.
When $\tau$ is too small, the confidence distribution becomes overly sharp, causing the model to place excessive weight on a small number of candidates, which may amplify overconfident but incorrect repairs. 
When $\tau$ is too large, the distribution becomes overly smooth, weakening the distinction between reliable and unreliable repairs.
We therefore use $\tau=1.0$ as the default setting.

\begin{figure}[t]
  \centering
  \includegraphics[width=\linewidth, trim=0 5 0 5, clip]{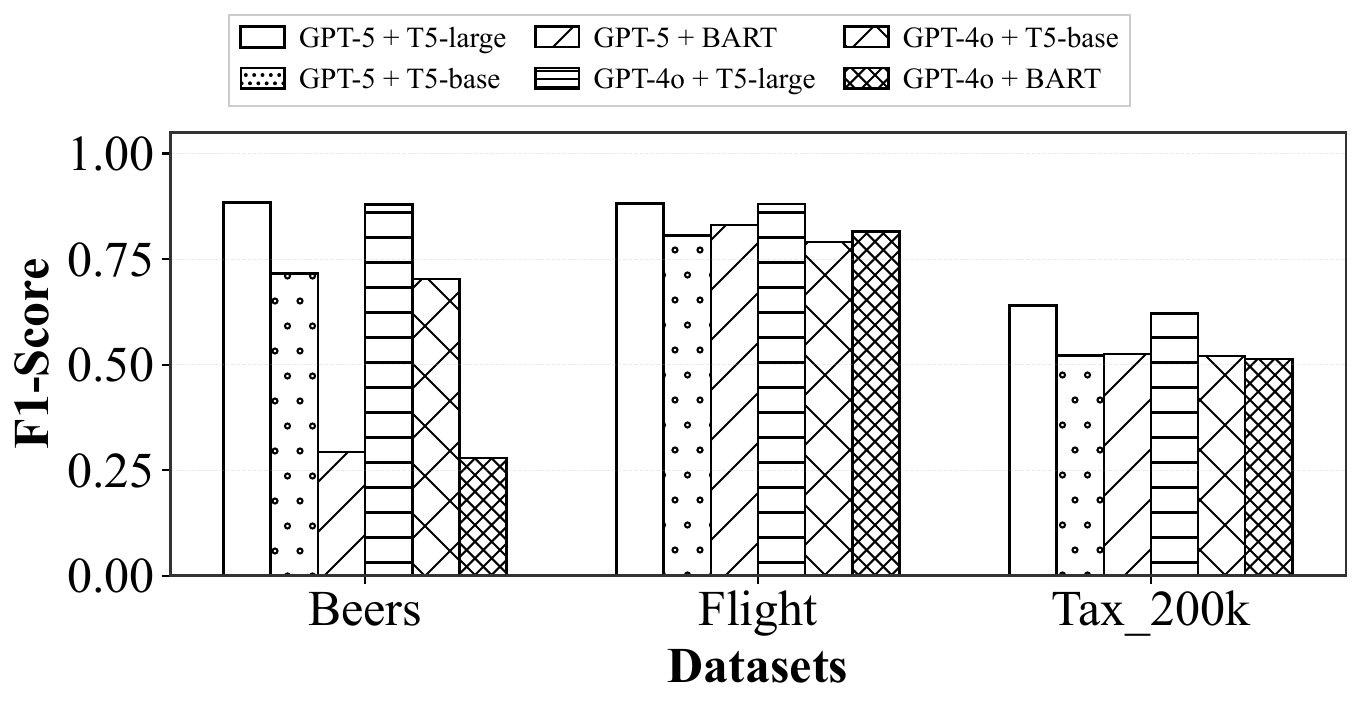}
  \caption{F1-score under different LLM--SLM combinations on Beers, Flight, and Tax\_200k. Larger correctors consistently yield better repair accuracy, while changing the instructor has a smaller effect.}
  \label{fig:exp4}
\end{figure}

\noindent\textbf{Exp-6: Effect of LLM--SLM combinations.}
We investigate how different instructor and corrector choices affect repair performance.
For the LLM instructor, we compare GPT-5 and GPT-4o~\cite{openai2024gpt4o}. 
For the SLM corrector, we vary model capacity by choosing from T5-large (0.77B parameters), T5-base (0.22B), and BART-base (0.1B)~\cite{lewis2019bart}.
We select the Beers, Flight, and Tax\_200k datasets because they vary in scale and include all error types.
Figure~\ref{fig:exp4} reports the resulting F1-scores.

The results show that the capacity of the corrector has a much stronger impact than the choice of the instructor.
With GPT-5 as the instructor, replacing T5-large with T5-base and BART-base reduces the average F1-score by 15.04\% and 31.45\%.
In contrast, replacing GPT-5 with GPT-4o while retaining T5-large reduces the average F1-score by 1.01\%.
These results indicate that the corrector capability has a substantially greater impact on repair performance than the instructor.
This trend is especially evident on more challenging datasets, where effective repair requires the corrector to learn complex value patterns and dependencies. 
In contrast, when the corrector is fixed, changing the instructor from GPT-5 to GPT-4o leads to minor differences. 
This suggests that once the instructor provides a reasonable column-level context structure, the dominant factor becomes the repair capacity of the corrector.
We therefore use GPT-5 as the instructor and T5-large as the corrector.
\begin{table*}[t]
    \centering
    \caption{Effect of training size on repair performance, evaluated by F1-score, EDR, and RDRR.}
    \label{tab:trainingsize}
    \begin{adjustbox}{max width=\textwidth}
    \begin{tabular}{lccc ccc ccc ccc ccc}
        \toprule
        \multirow{2}{*}{Dataset\textbackslash Training size}
            & \multicolumn{3}{c}{200}
            & \multicolumn{3}{c}{500}
            & \multicolumn{3}{c}{1000}
            & \multicolumn{3}{c}{2000}
            & \multicolumn{3}{c}{3000} \\
        \cmidrule(lr){2-4}\cmidrule(lr){5-7}\cmidrule(lr){8-10}\cmidrule(lr){11-13}\cmidrule(lr){14-16}
            & F1 & EDR & RDRR
            & F1 & EDR & RDRR
            & F1 & EDR & RDRR
            & F1 & EDR & RDRR
            & F1 & EDR & RDRR \\
        \midrule
        Tax\_200k
            & 0.3458 & 0.3474 & 0.4052
            & 0.5174 & 0.5198 & 0.6063
            & 0.6172 & 0.6201 & 0.7232
            & 0.6398 & 0.6428 & 0.7497
            & 0.6284 & 0.6313 & 0.7363 \\
        Shuttle
            & 0.5580 & 0.5525 & 0.5789
            & 0.5775 & 0.5718 & 0.5991
            & 0.7947 & 0.7869 & 0.8245
            & 0.8745 & 0.8659 & 0.9073
            & 0.8869 & 0.8782 & 0.9201 \\
        \bottomrule
    \end{tabular}
    \end{adjustbox}
\end{table*}

\begin{figure}[t]
  \centering
  \includegraphics[width=\linewidth, trim= 0 0 0 0]{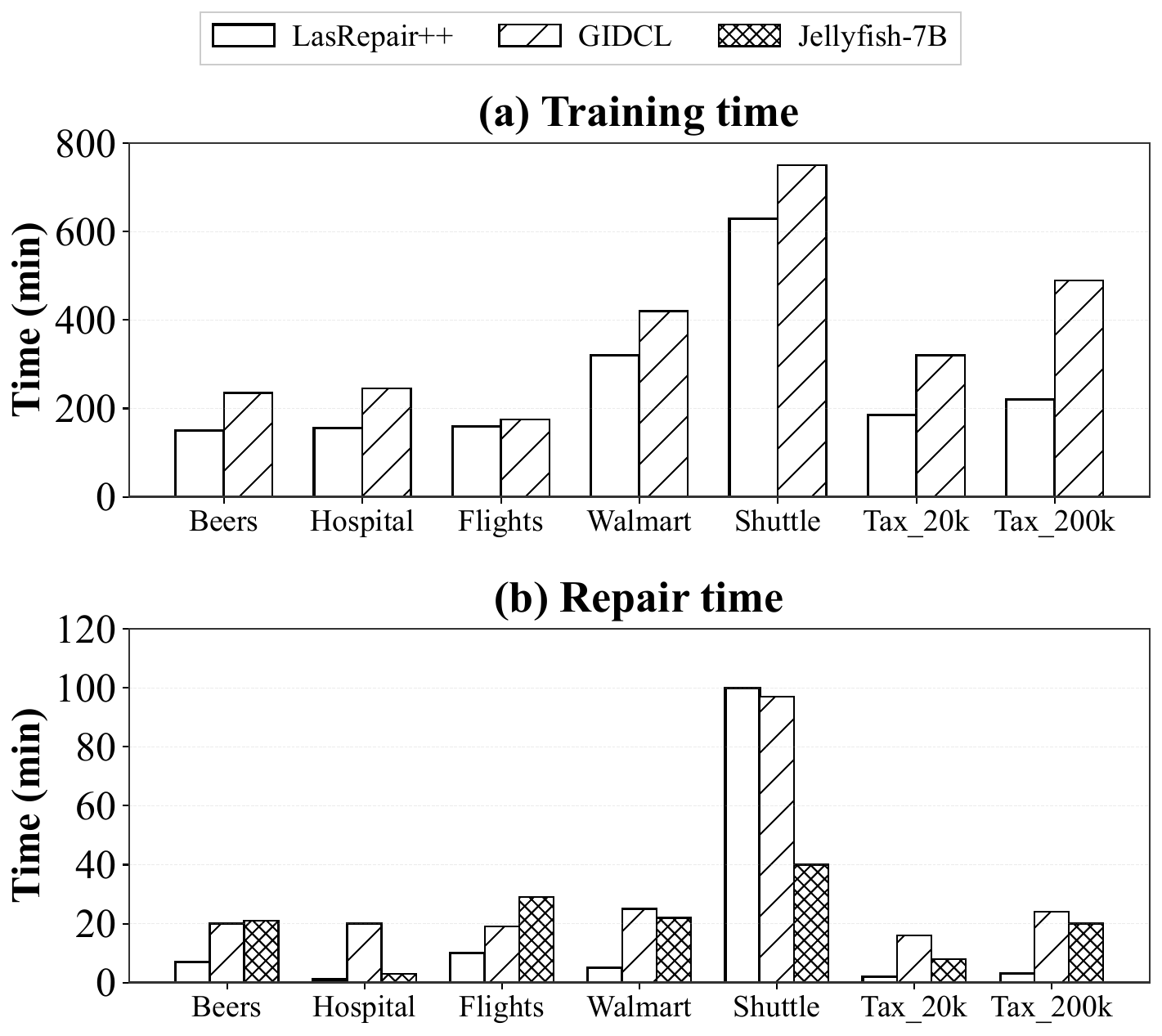}
  \caption{Running time of different repair methods on seven datasets. Panels (a) and (b) report training time and repair time, respectively, both in minutes.}
  \label{fig:runtime}
\end{figure}


\noindent\textbf{Exp-7: Effect of training set size.}
We evaluate the sample efficiency of \Frameworkname{}++ on large-scale datasets.
We vary the number of sampled training tuples on two large-scale datasets, Shuttle and Tax\_200k. 
Specifically, we train with $\{200, 500, 1{,}000, 2{,}000, 3{,}000\}$ tuples and report F1-score, EDR, and RDRR in Table~\ref{tab:trainingsize}.

Overall, performance improves as the training set size increases. 
With 200 tuples, the corrector is exposed to limited dataset-specific training data and therefore achieves relatively low accuracy.
Increasing the training set size to 1{,}000 tuples brings substantial gains, showing that \Frameworkname{}++ can learn useful repair behavior from a moderate number of examples.
Further increasing the training set size to 2{,}000--3{,}000 tuples provides additional but smaller improvements, suggesting diminishing returns.
This result supports our default choice of using 2{,}000 sampled tuples for large-scale datasets, which balances repair performance and cost.

\subsection{Efficiency evaluation}
\noindent\textbf{Exp-8: Runtime.}
We compare the runtime of \Frameworkname{}++ against the generative repair baselines, reporting both training time and repair time on 7 datasets in Figure~\ref{fig:runtime}. 

In the training stage, \Frameworkname{}++ fine-tunes a lightweight corrector rather than an LLM.
Its training time is lower than or comparable to the generative baselines on most datasets.
Compared with GIDCL, \Frameworkname{}++ reduces training time by approximately 40\% on average across the evaluated datasets.

In the repair stage, \Frameworkname{}++ is also consistently efficient because the instructor is used only once to select column-level context, while the corrector performs cell-level repair.
Therefore, the repair cost scales mainly with the number of detected erroneous cells and the corrector inference cost.
Across the seven datasets, \Frameworkname{}++ reduces repair time by 74\% on average compared with GIDCL and by about 69\% compared with Jellyfish.
The exception is Shuttle, where the high error rate results in a large number of detected cells and thus incurs repeated repair operations.
Overall, \Frameworkname{}++ achieves an efficiency--accuracy trade-off.

\begin{figure}[t]
  \centering
  \includegraphics[width=\linewidth]{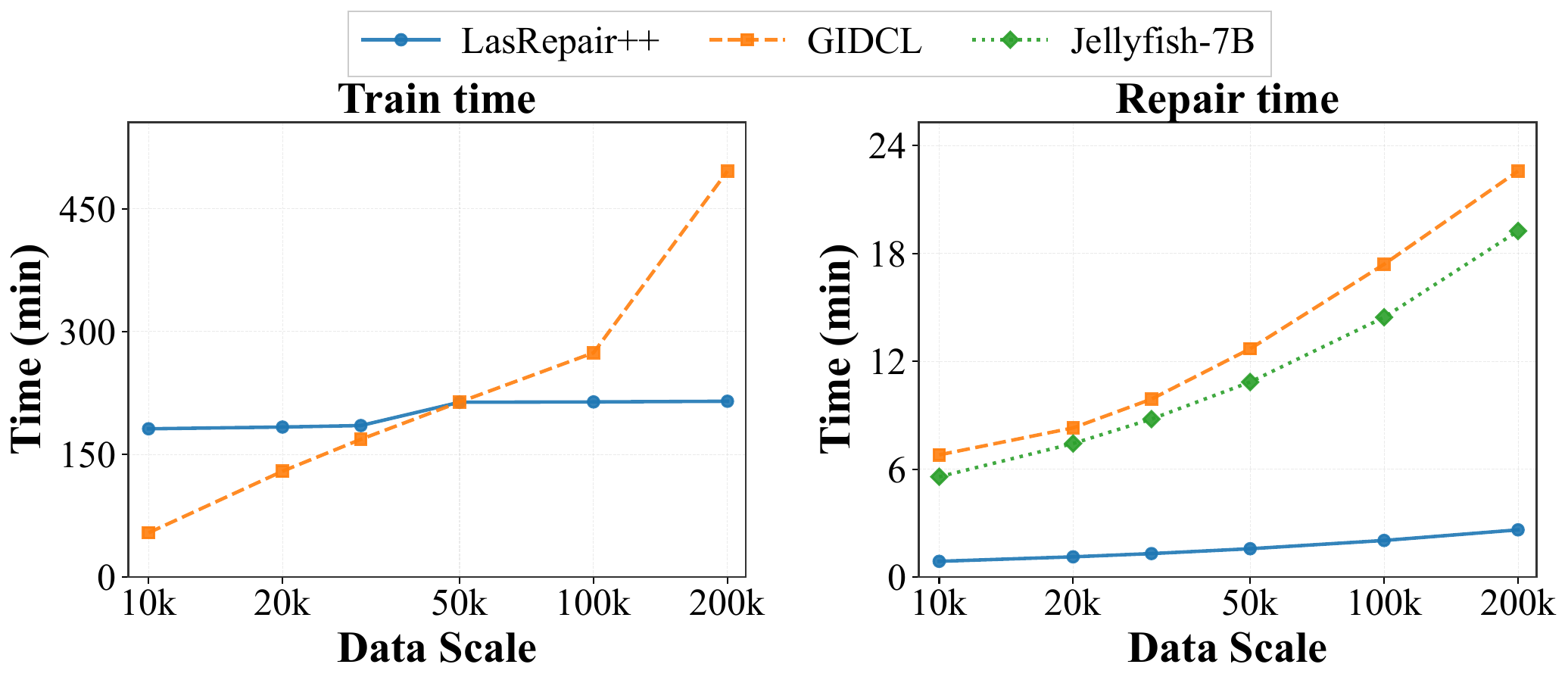}
  \caption{Scalability of different repair methods under increasing data scales on the Tax dataset.}
  \label{fig:scalability}
\end{figure}



\noindent\textbf{Exp-9: Scalability.}
Figure~\ref{fig:scalability} reports the training and repair time of \Frameworkname{}++, Jellyfish-7B, and GIDCL as the Tax dataset scales from 10k to 200k tuples.
Jellyfish-7B is omitted from the training comparison because it does not require dataset-specific fine-tuning.

For training, \Frameworkname{}++ exhibits a relatively flat growth trend because it fine-tunes the corrector on a bounded sample of tuples for large-scale datasets.
When the dataset size increases by $20\times$, its training time increases only from 181.26 to 214.91 minutes, corresponding to $1.19\times$ the original time.
In contrast, the training time of GIDCL grows by 9.15$\times$, from 54.19 to 496.13 minutes.

For repair, GIDCL and Jellyfish-7B exhibit similar trends because they use the same underlying LLM.
As the dataset size increases from 10k to 200k, the repair time of \Frameworkname{}++ grows by approximately $2.9\times$, whereas GIDCL and Jellyfish-7B grow by about $3.3\times$ and $3.4\times$, respectively. These results demonstrate the superior scalability of \Frameworkname{}++.

\subsection{Ablation study}
\begin{figure}[t]
  \centering
  \includegraphics[width=\linewidth]{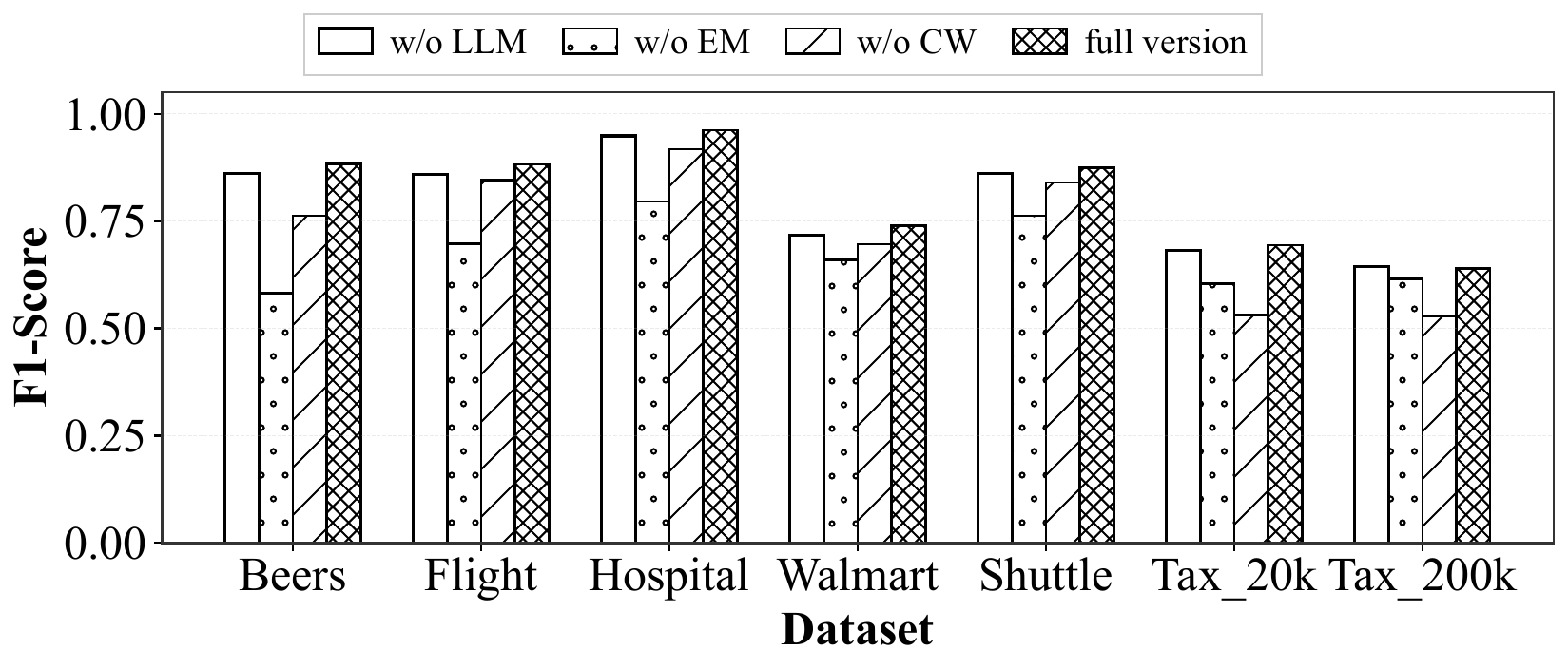}
  \caption{Ablation study of \Frameworkname{}++ on seven benchmark datasets. We individually remove each component, including LLM-based context selection (LLM), the EM-style procedure (EM), and confidence-based weighting (CW), and report the resulting F1-scores.}
  \label{fig:ablation}
\end{figure}
\noindent\textbf{Exp-10: Ablation study.}
We conduct an ablation study to quantify the contribution of three core components in \Frameworkname{}++:
(i) the LLM as an instructor for context selection, 
(ii) the EM-style procedure, and 
(iii) the confidence-based weighting. 
For each ablated variant, we remove one component while keeping the remaining settings unchanged.
Figure~\ref{fig:ablation} reports the F1-scores on 7 benchmark datasets. 

Removing the EM-style procedure causes the largest degradation: the average F1-score drops from 0.8104 to 0.6736, with especially large declines on Beers and Flight.
Removing confidence-based weighting also reduces the average F1-score to 0.7314, with the largest drops on Tax\_20k and Tax\_200k.
Removing the instructor results in an average F1-score decrease of 1.81\%, consistent with its primary role in filtering irrelevant attributes to improve inference efficiency.
Overall, the EM-style procedure makes the largest contribution, followed by confidence-based weighting and the instructor.
Taken together, the ablation results show that the three components of \Frameworkname{}++ provide complementary benefits.

\subsection{Case study}
\begin{figure}[t]
  \centering
  \includegraphics[width=\linewidth, trim=0 0 0 3, clip]{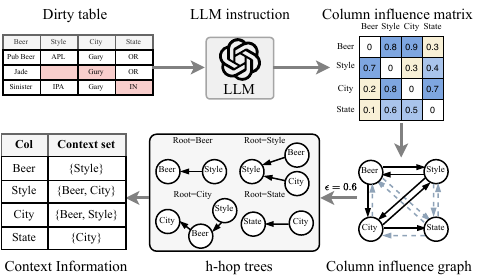}
  \caption{Case study demonstrating how the column-specific context sets are extracted from a dirty dataset.}
  \label{fig:case}
\end{figure}
We present a case study to illustrate the column-specific context selection of the LLM instructor. 
As shown in Figure~\ref{fig:case}, the instructor estimates inter-column influence from a simplified Beers dataset and extracts a context set for each target column from the resulting graph.
For example, although \textsf{Style} has no direct influence on \textsf{City} under the threshold $\epsilon=0.6$, its influence propagates through \textsf{Beer}. Therefore, both columns are selected as contexts for repairing \textsf{City}.
This example demonstrates that the instructor can extract compact and relevant contexts to support repair.


\section{Related Work}
\label{sec:related}
Existing methods can be broadly divided into \emph{conventional}, \emph{discriminative}, and \emph{generative} methods.

Conventional methods rely on integrity constraints or probabilistic inference. 
Constraint-oriented methods such as Holistic~\cite{xu2013holistic} and BigDansing~\cite{zuhair2015bigdansing} repair violations of integrity constraints, with BigDansing emphasizing scalability. 
Unified~\cite{fei2011unified} formulates repair under functional dependencies using a minimum-description-length principle.
HoloClean~\cite{theo2017holoclean} combines constraints, external knowledge, and statistical dependencies through probabilistic inference.

Discriminative methods learn predictive models to select appropriate repairs.
Baran~\cite{mohammad2020baran} generates candidates from contextual representations and uses column-specific classifiers for selection.
BoostClean~\cite{sanjay2017boostclean} learns combinations of detection and repair operations to improve downstream prediction.

Generative methods leverage pretrained language models to produce repairs from serialized table contexts.
Jellyfish~\cite{zhang2024jellyfish} adapts local LLMs for general-purpose data preprocessing tasks, including repair.
GIDCL~\cite{yan2024gidcl} combines graph-based table modeling with LLM-based repair.
Although these methods improve repair capability, they still depend on context quality and suffer from uncertain model outputs.
This limitation motivates our focus on instructor-guided context selection, EM-style procedure, and confidence-based weighting.

\section{Conclusion}
In this paper, we present the \Frameworkname{} family, which combines the semantic reasoning capability of LLMs with the efficiency of SLMs for scalable data repair. 
The base framework, \Frameworkname{}, uses an LLM as a table-level structural instructor to select compact and relevant repair contexts, and fine-tunes an SLM as a scalable cell-level corrector.
Building on this, \Frameworkname{}+ adopts an EM-style procedure to jointly refine the repaired table and the corrector over successive iterations.
\Frameworkname{}++ further applies confidence-based weighting to reduce the influence of unreliable generated repairs.
Experiments on 7 benchmark datasets highlight the effectiveness and efficiency of our method.

\clearpage
\balance
\bibliographystyle{IEEEtran}
\bibliography{sample}


\end{document}